\documentclass[aps,prl,reprint,superscriptaddress,nofootinbib,floatfix]{revtex4-2}
\usepackage{amsmath,amssymb,mathtools,bm,amsthm}
\usepackage{braket}
\usepackage{microtype}
\usepackage{enumitem}
\usepackage{xcolor,graphicx,tikz}
\usetikzlibrary{quantikz2,shapes.geometric}
\usepackage[hidelinks,hypertexnames=false]{hyperref}

\usepackage{orcidlink}
\makeatletter
\let\selectlanguage\@gobble
\makeatother

\allowdisplaybreaks[3]

\newcommand{\id}{\mathbf{1}}
\newcommand{\Pos}{\mathrm{Pos}}
\newcommand{\Herm}{\mathrm{Herm}}
\newcommand{\Tr}{\operatorname{Tr}}
\newcommand{\supp}{\operatorname{supp}}
\newcommand{\Ran}{\operatorname{Ran}}
\newcommand{\Ker}{\operatorname{Ker}}
\newcommand{\ECS}{\mathrm{ECS}}
\newcommand{\ECSplus}{\mathrm{ECS}^{+}}
\newcommand{\CS}{\mathrm{CS}}
\newcommand{\CSplus}{\mathrm{CS}^{+}}
\newcommand{\QCCC}{\mathrm{QC\mbox{-}CC}}
\newcommand{\pECS}{\mathrm{pECS}}
\newcommand{\pQCCC}{\mathrm{pQC\mbox{-}CC}}
\newcommand{\D}{\mathcal D}
\newcommand{\I}{\mathcal I}

\newcommand{\N}{\mathcal N}
\renewcommand{\proj}[1]{|#1\rangle\!\langle#1|}
\newcommand{\oprod}[2]{|#1\rangle\!\langle#2|}

\definecolor{tnred}{HTML}{D24A43}
\definecolor{tnredfill}{HTML}{FFF9F8}
\definecolor{tnblue}{HTML}{4E55CC}
\definecolor{tnbluefill}{HTML}{F0F0FF}
\definecolor{tnwire}{HTML}{666666}
\definecolor{tntext}{HTML}{111111}
\newtheorem{theorem}{Theorem}
\newtheorem{proposition}[theorem]{Proposition}
\newtheorem{lemma}[theorem]{Lemma}
\newtheorem{corollary}[theorem]{Corollary}
\newtheorem*{corollaryone}{Corollary 1}
\theoremstyle{definition}
\newtheorem{definition}[theorem]{Definition}

\begin{document}

\title{Extensibly Causally Separable Processes Admit Realizations as Quantum Circuits with Classical Control of Causal Order}
\author{Wenjie Wei\orcidlink{0009-0001-2435-0852}}
\affiliation{School of Physics, Sun Yat-sen University, Guangzhou, Guangdong 510275, China}

\author{Shengshi Pang\orcidlink{0000-0002-6351-539X}}
\email{pangshsh@mail.sysu.edu.cn}
\affiliation{School of Physics, Sun Yat-sen University, Guangzhou, Guangdong 510275, China}
\affiliation{Hefei National Laboratory, University of Science and Technology of China, Hefei 230088, China}
\date{\today}

\begin{abstract}

Indefinite causal order extends the conventional circuit paradigm by allowing local quantum operations to be connected without a predetermined global order. Within the process-matrix framework, it is characterized by causal nonseparability, whereas extensible causal separability (ECS) requires that, under arbitrary input-ancilla extensions, a process admit a recursive decomposition into components, each compatible with a particular operation acting first.
Quantum circuits with classical control of causal order (QC-CC), in which previous outcomes dynamically determine which operation acts next, are known to generate ECS processes, but whether every ECS process admits such a realization has remained a longstanding open problem.
We settle it through a generalized teleportation construction that realizes every multipartite ECS process as a QC-CC. We further establish consistency between ECS definitions for trivial and nontrivial global past and future systems. These results provide ECS with an exact operational interpretation and identify QC-CC as the complete circuit structure underlying this class of processes. 
\end{abstract}

\maketitle

\smallskip\noindent\textit{Introduction.\textemdash}
Causal order determines which events can influence one another and shapes how information can flow. Conventional quantum circuits encode this structure by arranging operations in a fixed order~\cite{chiribella2009TheoreticalFrameworkQuantum}. Quantum theory broadens this picture by allowing different orders of operations to be coherently controlled, as exemplified by the quantum switch~\cite{chiribella2008TransformingQuantumOperations,chiribella2013NormalCompletelyPositive,bisio2019TheoreticalFrameworkHigherorder,oreshkov2012QuantumCorrelationsNo,procopio2015ExperimentalSuperpositionOrders,rubino2017ExperimentalVerificationIndefinite,goswami2018IndefiniteCausalOrder,goswami2020ExperimentsQuantumCausality,chiribella2013QuantumComputationsDefinite}. The process-matrix framework goes further by describing local quantum operations without assuming a predefined global order~\cite{hardy2007QuantumGravityFramework,oreshkov2012QuantumCorrelationsNo,brukner2014QuantumCausality,adlam2020OperationalChoiJamiolkowski}. 
This framework provides a general language for indefinite causal order (ICO), but its abstract consistency conditions neither guarantee that a given abstract process admits a circuit realization nor identify the physical resource underlying such a process. Complementary approaches have studied purifiability and operational or spacetime realizations for specific process classes~
\cite{araujo2017PurificationPostulateQuantum,oreshkov2019TimedelocalizedQuantumSubsystems,salzger2025MappingIndefiniteCausal}.
Establishing correspondences between abstract process classes and operational architectures is thus essential for both understanding the physical significance of the process-matrix framework and identifying the resources underlying ICO phenomena~\cite{vanrietvelde2025ConsistentCircuitsIndefinite,grothus2025RoutingQuantumControl,vanrietvelde2021RoutedQuantumCircuits}.

The need for such a classification is reinforced by the potential advantages of ICO in quantum computation~\cite{araujo2014ComputationalAdvantageQuantumcontrolled,araujo2017QuantumComputationIndefinite,taddei2021ComputationalAdvantageQuantum}, channel discrimination and quantum metrology~\cite{bavaresco2021StrictHierarchyParallel,zhao2020QuantumMetrologyIndefinite,procopio2023ParameterEstimationIndefinite,liu2023OptimalStrategiesQuantum,li2026OptimalStrategiesMultiparameter}, and communication through noisy channels~\cite{feix2015QuantumSuperpositionOrder,ebler2018EnhancedCommunicationAssistance,procopio2019CommunicationEnhancementQuantum,chiribella2021IndefiniteCausalOrder}. Nevertheless, an ICO advantage may depend on coherent path control or the structure of the channels, rather than on causal nonseparability alone~\cite{guerin2019CommunicationQuantumcontrolledNoise,abbott2020CommunicationCoherentControl}. Identifying the circuit structure behind a process is therefore important for determining which resource is actually responsible for an observed advantage.

Within the process-matrix framework, causal nonseparability characterizes ICO~\cite{araujo2015WitnessingCausalNonseparability}. A causally separable process admits a recursive decomposition into positive semidefinite components, with each component compatible with a particular operation acting first~\cite{ognyanoreshkov2016CausalCausallySeparable,wechs2019DefinitionCharacterisationMultipartite}. The recursion allows subsequent operations to depend on earlier outcomes, so the complete order need not be selected before the experiment~\cite{ognyanoreshkov2016CausalCausallySeparable,milz2021DelayedchoiceCausalOrder,mothe2025CorrelationsQuantumCircuits}. Extensible causal separability (ECS) requires such a recursive decomposition under arbitrary shared input-ancilla extensions~\cite{ognyanoreshkov2016CausalCausallySeparable,wechs2019DefinitionCharacterisationMultipartite}. Although this definition accommodates dynamically chosen causal orders, it does not specify how an ECS process can be assembled as a circuit.

Quantum circuits with classical control of causal order (QC-CC) provide a candidate realization~\cite{wechs2021QuantumCircuitsClassical,wechs2023ErratumQuantumCircuits}. In these circuits, instrument outcomes determine which unused operation is applied next. Each branch has a well-defined order, while different outcomes may generate different orders dynamically. Every QC-CC process is known to be ECS, but whether QC-CC exhausts the general multipartite ECS class has remained a longstanding open problem~\cite{ognyanoreshkov2016CausalCausallySeparable,wechs2019DefinitionCharacterisationMultipartite,wechs2021QuantumCircuitsClassical}. A negative answer would reveal ECS processes that cannot be generated by any dynamically ordered circuit with classical control. A positive answer would show that the abstract ECS condition already identifies a complete operational architecture.

In this work, we resolve this open problem by proving that $\mathrm{ECS}=\mathrm{QC\text{-}CC}$. Using a generalized teleportation construction, we provide an explicit QC-CC circuit realization for every ECS process, thereby yielding a finite semidefinite-programming characterization of ECS. We further show that this equivalence is independent of whether the global past and future are treated as terminal systems or as additional parties. This establishes consistency between the definitions of ECS with trivial and nontrivial global past and future systems. 
We also define probabilistic ECS and establish its equivalence with probabilistic QC-CC.

\smallskip\noindent\textit{ECS and QC-CC.\textemdash}
Consider a quantum process involving $N$ local laboratories, each implementing one external operation, with label set $\N=\{1,\ldots,N\}$. The input and output spaces associated with operation $k$ are $A_I^k,A_O^k$, respectively, and we write $A_{IO}^k:=A_I^k\otimes A_O^k$. For $\mathcal S\subseteq\N$, let $A_{IO}^{\mathcal S}:=\bigotimes_{k\in\mathcal S}A_{IO}^k$. Juxtaposed system labels denote tensor products. The global past and future are denoted by $P,F$, respectively, and $\Pos(\mathcal H)$ is the cone of PSD operators on the Hilbert space $\mathcal H$.

A deterministic $P\to F$ process matrix is a PSD operator on $P\otimes A_{IO}^{\N}\otimes F$, denoted by $W$, such that inserting arbitrary completely positive and trace-preserving (CPTP) maps into the parties' input-output systems yields a CPTP $P\to F$ map. Its validity conditions comprise positive semidefiniteness, homogeneous linear equality constraints, and an affine normalization constraint~\cite{chiribella2009TheoreticalFrameworkQuantum,wechs2021QuantumCircuitsClassical,araujo2015WitnessingCausalNonseparability}. Writing $D_O(\N):=\prod_{k\in\N}\dim A_O^k$, the last constraint is $\Tr_{A_{IO}^{\N}F}W=D_O(\N)\id^P$. Removing this affine constraint while retaining the first two conditions defines the cone of unnormalized processes, denoted by $\mathsf{Proc}_N^{+,P\to F}$.
Equivalently, $X\in\mathsf{Proc}_N^{+,P\to F}$ if and only if $X\succeq0$ and the Choi--Jamiołkowski (CJ) operators $M_k$ of arbitrary local CPTP maps satisfy $\Tr_F[(\bigotimes_{k\in\N}M_k)*X]=\Theta_X$, where $\Theta_X:=D_O(\N)^{-1}\Tr_{A_{IO}^{\N}F}X$. The positive operator $\Theta_X$ is fixed by $X$ and is independent of the inserted CPTP maps. Deterministic normalization is the condition $\Theta_X=\id^P$. The homogeneous linear constraints and the proof of this equivalence are given in the Supplemental Material (SM), Sec.~\ref{sec:foundations}.

To define ECS, we first introduce conditioning on a local operation. We add a finite-dimensional input ancilla to party $k$, denoted by $E_k$. Let $X$ be an operator on the extended process space $P\otimes A_I^kE_kA_O^k\otimes A_{IO}^{\N\setminus\{k\}}\otimes F$. A quantum operation by this party has a CJ matrix on $A_I^k\otimes E_k\otimes A_O^k$, denoted by $C_k$, satisfying $C_k\succeq0$ and $\Tr_{A_O^k}C_k\preceq\id^{A_I^kE_k}$. Conditioning $X$ on this quantum operation is defined by $X_{|C_k}:=
\Tr_{A_I^kE_kA_O^k}\!\left[
\left(C_k\otimes\id\right)X
\right].$ If $X$ is a valid process and, for party $k$ and every quantum operation $C_k$, the conditional process belongs to the cone of unnormalized valid processes for the remaining parties, then $X$ is compatible with party $k$ acting first~\cite{wechs2019DefinitionCharacterisationMultipartite}. The separability definition imposes this compatibility recursively through conditioning.

The multipartite definition of Wechs et al.~\cite{wechs2019DefinitionCharacterisationMultipartite} uses CS for the extensible notion described in the Introduction. We follow the terminology of Oreshkov and Giarmatzi: CS denotes the unextended recursive definition, whereas ECS requires it to hold under extension.

We define CS processes with nontrivial $P$ and $F$. Set $\CS_0^{+,P\to F}:=\Pos(P\otimes F)$. For $N\ge1$, we have $X\in\CS_N^{+,P\to F}$ if and only if there exist operators $X_{(k)}\in\mathsf{Proc}_N^{+,P\to F}$ such that, for all $k\in\N$ and all quantum operations $C_k$~\cite{ognyanoreshkov2016CausalCausallySeparable,wechs2019DefinitionCharacterisationMultipartite},
\begin{equation}\label{eq:cs-recursion}
X=\sum_{k\in\N}X_{(k)},\;
(X_{(k)})_{|C_k}\in\CS_{N-1}^{+,P\to F}
\end{equation}
This recursion allows the order of subsequent operations to depend on previously performed operations and their classical outcomes. Multipartite CS processes can therefore extend beyond convex mixtures obtained by randomly selecting a fixed permutation before the experiment~\cite{ognyanoreshkov2016CausalCausallySeparable,wechs2019DefinitionCharacterisationMultipartite,milz2021DelayedchoiceCausalOrder,mothe2025CorrelationsQuantumCircuits}.

For any system $R$, let $d_R:=\dim R$ and define the completely depolarizing superoperator $\D_R(Y):=\Tr_R(Y)\otimes\id^R/d_R$. The identity superoperator is denoted by $\I$. The first relation below follows directly from $X_{(k)}\in\mathsf{Proc}_N^{+,P\to F}$. The remaining relations follow from the recursive condition, as shown in SM, Proposition~\ref{prop:pf-first-component-order}. Each $X_{(k)}$ is therefore compatible with the order $P\prec k\prec(\N\setminus\{k\})\prec F$. Explicitly, for every nonempty $\mathcal X\subseteq\N\setminus\{k\}$, this compatibility requires
\begin{align}
&(\I-\D_{A_O^k})\D_{A_{IO}^{\N\setminus\{k\}}F}X_{(k)}=0,\nonumber\\
&\left(\prod_{i\in\mathcal X}(\I-\D_{A_O^i})\right)
\D_{A_{IO}^{(\N\setminus\{k\})\setminus\mathcal X}F}X_{(k)}=0.
\label{eq:main-first-order-compatibility}
\end{align}

Extensible causal separability rules out the activation of noncausality by entangled input ancillas and is therefore stronger than unextended causal separability. Set $\ECS_0^{+,P\to F}:=\Pos(P\otimes F)$. For an operator $X$ on the original process space, define $X\in\ECS_N^{+,P\to F}$ if and only if, after adding arbitrary finite-dimensional input ancillas $E_1,\ldots,E_N$ to the respective parties, every normalized shared state $\rho$ satisfies~\cite{ognyanoreshkov2016CausalCausallySeparable,wechs2019DefinitionCharacterisationMultipartite}
\begin{equation}
X\otimes\rho\in\CS_N^{+,P\to F}.
\label{eq:ecs-extension}
\end{equation}
Here the right-hand side is understood on the extended input spaces. The deterministic section is $\ECS_N^{P\to F}:=\{X\in\ECS_N^{+,P\to F}:\Theta_X=\id^P\}$.
QC-CC has a direct circuit definition~\cite{wechs2021QuantumCircuitsClassical,wechs2023ErratumQuantumCircuits}. A deterministic process is QC-CC if and only if it can be realized using quantum instruments and their classical outcomes. The initial instrument determines which external operation is applied first. After each incomplete history, an intermediate instrument uses the history of previously performed operations and their outcomes to determine which unused external operation is applied next. After all external operations have been applied, the final map outputs to $F$. The deterministic processes admitting such circuit realizations form $\QCCC_N^{P\to F}$.

This operational definition is equivalent to the following characterization. For a finite label set $\mathcal T$, let $\Pi_{\mathcal T}$ denote its set of full permutations. A history $h=(k_1,\ldots,k_n)$ is an ordered sequence of distinct labels. The notation $\pi\succeq h$ means that the permutation $\pi$ begins with $h$. The empty sequence is a prefix of every permutation. A deterministic process $W$ belongs to $\QCCC_N^{P\to F}$ if and only if there exist PSD operators $\{T_\pi\}_{\pi\in\Pi_{\N}}$ such that, for every nonempty history $h$,
\begin{equation}
\begin{aligned}
T_h:=\sum_{\pi\succeq h}T_\pi,
W=\sum_{\pi\in\Pi_{\N}}T_\pi.
\end{aligned}
\label{eq:qccc-decomposition}
\end{equation}
\begin{equation}
(\I-\D_{A_O^{k_n}})\D_{A_{IO}^{\N\setminus\{k_1,\ldots,k_n\}}F}T_h=0.
\label{eq:qccc-sequence-constraint}
\end{equation}
Here $T_h$ sums over all complete orders extending history $h$, and Eq.~\eqref{eq:qccc-sequence-constraint} expresses trace preservation of the corresponding internal operations after summing over their classical outcomes. The next external operation can thus be determined dynamically from the history of previously performed operations~\cite{abbott2020CommunicationCoherentControl,mothe2025CorrelationsQuantumCircuits}. The equivalence between this operator characterization and a circuit realization is proved in SM, Proposition~\ref{prop:basic-pf-qccc}.

\smallskip\noindent\textit{Main theorem.\textemdash}For every $N\ge1$, finite-dimensional $P,F$, and deterministic process $W$,
$W\in\ECS_N^{P\to F}$ if and only if $W\in\QCCC_N^{P\to F}.$ The detailed proof is in SM, Sec.~\ref{sec:pf}.

The extension-independent first-level decomposition of Wechs et al.~\cite{wechs2019DefinitionCharacterisationMultipartite} also holds for fixed nontrivial $P,F$:

\begin{proposition}[First-level decomposition]
\label{prop:uniform-first-decomposition}
Let $N\ge1$ and $X\in\mathsf{Proc}_N^{+,P\to F}$. Then $X\in\ECS_N^{+,P\to F}$ if and only if there exists a family
$\{X_{(k)}\}_{k\in\N}\subset\mathsf{Proc}_N^{+,P\to F}$ satisfying
\begin{align}
X&=\sum_{k\in\N}X_{(k)},
\label{eq:uniform-first-decomposition}\\
(X_{(k)}\otimes\rho)_{|C_k}
&\in\CS_{N-1}^{+,P\to F}.
\label{eq:uniform-first-condition}
\end{align}
The second condition holds for every finite-dimensional input extension, every normalized shared state $\rho$, every $k$, and every quantum operation by party $k$ with CJ matrix $C_k$.
\end{proposition}
We first prove necessity, $W\in\ECS_N^{P\to F}\Rightarrow W\in\QCCC_N^{P\to F}$. Induction on $N$ establishes the stronger unnormalized statement: every $X\in\ECS_N^{+,P\to F}$ admits PSD operators indexed by full permutations $\pi$, denoted by $\{T_\pi\}_{\pi\in\Pi_{\N}}$, such that $X=\sum_\pi T_\pi$ and Eq.~\eqref{eq:qccc-sequence-constraint} holds. For $N=1$, take the sole component to be $T_{(k)}:=X$. Henceforth assume $N\ge2$.

For one-dimensional $P,F$, Wechs et al.~\cite{wechs2019DefinitionCharacterisationMultipartite} used teleportation to obtain a necessary condition for ECS: a first-level component in Eq.~\eqref{eq:uniform-first-decomposition}, denoted by $X_{(k)}$, remains ECS after its input-output system $A_{IO}^k$ is transferred to an input ancilla of any remaining party, now viewed as a process for the remaining $N-1$ parties. The resulting PSD decomposition may depend on which party receives the transferred system. In the general multipartite case, Eq.~\eqref{eq:qccc-decomposition} instead requires one family of components corresponding to complete orders whose prefix sums satisfy all history constraints. We use a isometry to coherently encode $A_{IO}^k$ into orthogonal input sectors of all candidate remaining parties. This generalization of teleportation converts the constraints supplied by the induction hypothesis into those required by Eq.~\eqref{eq:qccc-sequence-constraint}.

For each first-level component $X_{(k)}$, we encode the input-output system of party $k$ into orthogonal sectors of the remaining parties. A shared input state and a maximally entangled projection realize this encoding. The resulting ECS process has $N-1$ parties, so its inductive decomposition can be compressed back to the original process space. By Proposition~\ref{prop:uniform-first-decomposition}, choose first-level components satisfying Eqs.~\eqref{eq:uniform-first-decomposition}--\eqref{eq:uniform-first-condition}. Fix $k\in\N$ and let $d_k:=\dim A_{IO}^k$.
 Fix an orthonormal basis of $A_{IO}^k$, denoted by $\{\ket a_{A_{IO}^k}\}_{a=1}^{d_k}$. For each $r\in\N\setminus\{k\}$, choose a basis $\{\ket a_{S_r}\}_{a=1}^{d_k}$ of $S_r\simeq A_{IO}^k$ and set $G_r:=\mathbb C\ket{\mathrm{vac}}_{G_r}\oplus S_r$. In this construction, $S_r$ is the data sector and $\mathbb C\ket{\mathrm{vac}}_{G_r}$ is the orthogonal vacuum sector. On each basis vector, define
\begin{equation}
V_r\ket a_{A_{IO}^k}:=\ket a_{S_r}
\otimes\!\bigotimes_{s\in\N\setminus\{k,r\}}
\ket{\mathrm{vac}}_{G_s}.
\label{eq:route-isometry}
\end{equation}
This defines an isometry $V_r:A_{IO}^k\to\bigotimes_{s\in\N\setminus\{k\}}G_s$. Orthogonality of the data and vacuum sectors gives $V_r^\dagger V_s=\delta_{rs}\id^{A_{IO}^k}$, so $V:=(N-1)^{-1/2}\sum_{r\in\N\setminus\{k\}}V_r$
satisfies $V^\dagger V=\id^{A_{IO}^k}$. Each $V_r$ corresponds to a candidate next party, while their normalized sum $V$ embeds all candidates into the same remaining input space.

To realize this encoding as a shared input before conditioning, introduce two independent ancillary copies, $\widetilde A_{IO}^k\simeq A_{IO}^k$ and $S\simeq A_{IO}^k$, with orthonormal bases corresponding to that of the actual slot. Transport $V$ to the ancillary copy through this fixed basis identification and denote the resulting isometry by $\widetilde V:\widetilde A_{IO}^k\to\bigotimes_{r\in\N\setminus\{k\}}G_r$. The symbol $V$ continues to denote the isometry on the actual slot $A_{IO}^k$. Prepare a maximally entangled state on $S\widetilde A_{IO}^k$ and apply $\widetilde V$ to the ancillary copy. Let
$\ket{\Phi_{d_k}}_{S\widetilde A_{IO}^k}:=d_k^{-1/2}\sum_{a=1}^{d_k}\ket a_S\ket a_{\widetilde A_{IO}^k}$. In the corresponding bases, the same notation also denotes the maximally entangled state between $S$ and the actual slot $A_{IO}^k$.
The shared ancillary state is $\ket{\Omega_V}:=(\id^S\otimes\widetilde V)\ket{\Phi_{d_k}}_{S\widetilde A_{IO}^k}\in S\otimes\bigotimes_{r\in\N\setminus\{k\}}G_r$.

Conditioning at party $k$ implements the isometric encoding of $X_{(k)}$. Write $\mathcal V:=\id^P\otimes V\otimes\id^{A_{IO}^{\N\setminus\{k\}}F}$. At party $k$, choose a quantum operation that projects the actual slot $A_{IO}^k$ and input ancilla $S$ onto a maximally entangled state, with CJ matrix $C_k:=\lambda\proj{\Phi_{d_k}}_{S A_{IO}^k}$, where $0<\lambda\le1$, and set $\beta:=\lambda/d_k^2>0$. As shown in Fig.~\ref{fig:teleport-isometry}, conditioning on this maximally entangled projection yields
\begin{equation}
Z_k:=\bigl(X_{(k)}\otimes\proj{\Omega_V}\bigr)_{|C_k}
=\beta\mathcal V X_{(k)}\mathcal V^\dagger .
\label{eq:teleport}
\end{equation}

For arbitrary further finite-dimensional input ancillas $H_r$ and any normalized state $\tau$ on $\bigotimes_{r\in\N\setminus\{k\}}H_r$, set $E_k=S$, $E_r=G_r\otimes H_r$, and $\rho=\proj{\Omega_V}\otimes\tau$ in Eq.~\eqref{eq:uniform-first-condition}. This gives $Z_k\otimes\tau\in\CS_{N-1}^{+,P\to F}$, so $Z_k\in\ECS_{N-1}^{+,P\to F}$ on inputs $A_I^rG_r$ by Eq.~\eqref{eq:ecs-extension}.

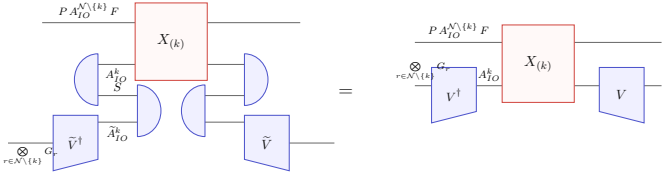
\begin{figure}[t]
\centering
\resizebox{\columnwidth}{!}{%
\begin{tikzpicture}[x=1cm,y=1cm,line cap=butt,line join=round]

\path[use as bounding box] (0.25,1.28) rectangle (17.47,5.70);

\begin{scope}[draw=tnwire,line width=0.675pt]
  \draw (1.15,5.18) -- (3.60,5.18);
  \draw (5.50,5.18) -- (7.95,5.18);

  \draw (2.62,4.08) -- (3.60,4.08);
  \draw (5.50,4.08) -- (6.48,4.08);

  \draw (2.62,3.26) -- (3.66,3.26);
  \draw (5.44,3.26) -- (6.48,3.26);

  \draw (2.62,2.56) -- (3.66,2.56);
  \draw (5.44,2.56) -- (6.48,2.56);

  \draw (0.25,2.01) -- (1.44,2.01);
  \draw (7.66,2.01) -- (8.85,2.01);

  \draw (10.97,4.66) -- (13.27,4.66);
  \draw (15.17,4.66) -- (17.47,4.66);

  \draw (10.97,3.52) -- (11.41,3.52);
  \draw (12.59,3.52) -- (13.27,3.52);
  \draw (15.17,3.52) -- (15.83,3.52);
  \draw (17.01,3.52) -- (17.47,3.52);
\end{scope}

\path[draw=tnred,fill=tnredfill,line width=0.775pt,line join=miter]
  (3.60,3.66) rectangle (5.50,5.70);
\node[inner sep=0pt,text=tntext] at (4.55,4.68)
  {\fontsize{11}{12}\selectfont $X_{(k)}$};

\path[draw=tnblue,fill=tnbluefill,line width=0.775pt,line join=round]
  (2.62,4.36)
  -- (2.62,3.00)
  .. controls (2.2775834551,3.00) and (2.00,3.3044463701) .. (2.00,3.68)
  .. controls (2.00,4.0555536299) and (2.2775834551,4.36) .. (2.62,4.36)
  -- cycle;

\path[draw=tnblue,fill=tnbluefill,line width=0.775pt,line join=round]
  (3.66,3.54)
  -- (3.66,2.18)
  .. controls (4.0024165449,2.18) and (4.28,2.4844463701) .. (4.28,2.86)
  .. controls (4.28,3.2355536299) and (4.0024165449,3.54) .. (3.66,3.54)
  -- cycle;

\path[draw=tnblue,fill=tnbluefill,line width=0.775pt,line join=round]
  (5.44,3.54)
  -- (5.44,2.18)
  .. controls (5.0975834551,2.18) and (4.82,2.4844463701) .. (4.82,2.86)
  .. controls (4.82,3.2355536299) and (5.0975834551,3.54) .. (5.44,3.54)
  -- cycle;

\path[draw=tnblue,fill=tnbluefill,line width=0.775pt,line join=round]
  (6.48,4.36)
  -- (6.48,3.00)
  .. controls (6.8224165449,3.00) and (7.10,3.3044463701) .. (7.10,3.68)
  .. controls (7.10,4.0555536299) and (6.8224165449,4.36) .. (6.48,4.36)
  -- cycle;

\path[draw=tnblue,fill=tnbluefill,line width=0.775pt,line join=miter]
  (1.44,2.74) -- (2.62,2.74) -- (2.62,1.58) -- (1.44,1.28) -- cycle;
\node[inner sep=0pt,text=tntext] at (2.03,2.0208)
  {\fontsize{9.5}{10.5}\selectfont $\widetilde V^\dagger$};

\path[draw=tnblue,fill=tnbluefill,line width=0.775pt,line join=miter]
  (6.48,2.74) -- (7.66,2.74) -- (7.66,1.28) -- (6.48,1.58) -- cycle;
\node[inner sep=0pt,text=tntext] at (7.07,2.0208)
  {\fontsize{10}{11}\selectfont $\widetilde V$};

\node[inner sep=0pt,text=tntext] at (9.15,3.35)
  {\fontsize{15.5}{16.5}\selectfont $=$};

\path[draw=tnred,fill=tnredfill,line width=0.775pt,line join=miter]
  (13.27,3.08) rectangle (15.17,5.12);
\node[inner sep=0pt,text=tntext] at (14.22,4.10)
  {\fontsize{11}{12}\selectfont $X_{(k)}$};

\path[draw=tnblue,fill=tnbluefill,line width=0.775pt,line join=miter]
  (11.41,4.00) -- (12.59,4.00) -- (12.59,2.84) -- (11.41,2.54) -- cycle;
\node[inner sep=0pt,text=tntext] at (12.00,3.2808)
  {\fontsize{9.5}{10.5}\selectfont $V^\dagger$};

\path[draw=tnblue,fill=tnbluefill,line width=0.775pt,line join=miter]
  (15.83,4.00) -- (17.01,4.00) -- (17.01,2.54) -- (15.83,2.84) -- cycle;
\node[inner sep=0pt,text=tntext] at (16.42,3.2808)
  {\fontsize{10}{11}\selectfont $V$};

\node[above,font=\fontsize{8}{9}\selectfont,align=center]
  at (2.375,5.18) {$P\,A_{IO}^{\N\setminus\{k\}}\,F$};

\node[below,font=\fontsize{8}{9}\selectfont,align=center]
  at (3.11,4.08) {$A_{IO}^k$};

\node[above,font=\fontsize{9}{10}\selectfont] at (3.14,3.26) {$S$};

\node[below,font=\fontsize{8}{9}\selectfont,align=center]
  at (3.14,2.56) {$\widetilde A_{IO}^k$};

\node[below,font=\fontsize{7}{8}\selectfont,align=center]
  at (0.845,2.01) {$\mathop{\bigotimes}\limits_{r\in\N\setminus\{k\}}G_r$};

\node[above,font=\fontsize{8}{9}\selectfont,align=center]
  at (12.12,4.66) {$P\,A_{IO}^{\N\setminus\{k\}}\,F$};

\node[above,font=\fontsize{7}{8}\selectfont,align=center]
  at (11.19,3.52) {$\mathop{\bigotimes}\limits_{r\in\N\setminus\{k\}}G_r$};
\node[above,font=\fontsize{8}{9}\selectfont,align=center]
  at (12.93,3.52) {$A_{IO}^k$};

\end{tikzpicture}%
}
\caption{Isometric encoding and conditioning on a maximally entangled projection. 
}
\label{fig:teleport-isometry}
\end{figure}

Apply the induction hypothesis to $Z_k$ for the remaining $N-1$ parties and prepend $k$ to each permutation of those parties. Denote the resulting PSD components by $Z_\pi$, with $\pi\in\Pi_\N$ and $\pi\succeq(k)$. 

Since $0\preceq Z_\pi\preceq Z_k$, positivity implies $\Ker Z_k\subseteq\Ker Z_\pi$, hence $\supp Z_\pi\subseteq\supp Z_k\subseteq\Ran\mathcal V$. Every component can therefore be compressed through the same isometry.

We can therefore define $X_\pi:=\beta^{-1}\mathcal V^\dagger Z_\pi\mathcal V\succeq0$. Since $\mathcal V^\dagger\mathcal V=\id$, this recovery is unique and preserves the sum of components:
\begin{equation}
Z_k=\sum_{\pi\succeq(k)}Z_\pi ,
Z_\pi=\beta\mathcal VX_\pi\mathcal V^\dagger,
X_{(k)}=\sum_{\pi\succeq(k)}X_\pi.
\label{eq:permutation-pullback}
\end{equation}
The following identity converts an extended-space history constraint into the corresponding constraint on the original process space. Let $h=(r_1,\ldots,r_n)$ be any nonempty ordered sequence of distinct labels from $\N\setminus\{k\}$, and write $\mathcal V_{r_1}:=\id^P\otimes V_{r_1}\otimes\id^{A_{IO}^{\N\setminus\{k\}}F}$.
For any $Y\in\Herm(P\otimes A_{IO}^{\N}\otimes F)$, orthogonality of the data and vacuum sectors gives the identity
\begin{equation}
\begin{aligned}
&\mathcal V_{r_1}^\dagger
(\I-\D_{A_O^{r_n}})
\D_{\left(\bigotimes_{j\in\N\setminus\{k,r_1,\ldots,r_n\}}
A_{IO}^jG_j\right)F}
(\mathcal VY\mathcal V^\dagger)
\mathcal V_{r_1}\\
&=\frac{\prod_{j\in\N\setminus\{k,r_1,\ldots,r_n\}}d_{G_j}^{-1}}{N-1}
(\I-\D_{A_O^{r_n}})\D_{A_{IO}^{\N\setminus\{k,r_1,\ldots,r_n\}}F}Y .
\end{aligned}
\label{eq:isometric-sequence}
\end{equation}
The factor in Eq.~\eqref{eq:isometric-sequence} is strictly positive. Define $X_{(k,h)}:=\sum_{\pi\succeq(k,h)}X_\pi$. Substituting $Y=X_{(k,h)}$ into Eq.~\eqref{eq:isometric-sequence} and using the extended-space constraints yields
\begin{equation}
(\I-\D_{A_O^{r_n}})\D_{A_{IO}^{\N\setminus\{k,r_1,\ldots,r_n\}}F}X_{(k,h)}=0.
\label{eq:original-sequence}
\end{equation}
Equation~\eqref{eq:original-sequence} provides the constraint for the prefix $(k,r_1,\ldots,r_n)$. For each $\pi\succeq(k)$, define $T_\pi:=X_\pi$, so that $\sum_{\pi\succeq(k)}T_\pi=X_{(k)}$. The single-party output constraint for $X_{(k)}\in\mathsf{Proc}_N^{+,P\to F}$ gives Eq.~\eqref{eq:qccc-sequence-constraint} for the length-one prefix $(k)$, while Eq.~\eqref{eq:original-sequence} covers every longer prefix. Combining the components over all $k\in\N$ reconstructs $X$ and satisfies Eq.~\eqref{eq:qccc-sequence-constraint} for every nonempty ordered sequence. Thus the construction proves the stronger unnormalized complete-order decomposition statement for every $X\in\ECS_N^{+,P\to F}$. Applying it to the deterministic process $W$ under consideration, the operator characterization of QC-CC gives $W\in\QCCC_N^{P\to F}$, proving necessity. The proof of sufficiency is given in SM, Theorem~\ref{thm:pf-ECS-QCCC}.

\smallskip\noindent\textit{Consistency with ECS with trivial past and future.\textemdash} 
The past and future can be represented as two parties~\cite{wechs2021QuantumCircuitsClassical,costa2026IndefiniteQuantumCausality}: view $P$ as a past party $\mathsf P$ whose input is one-dimensional and whose output is $P$, and view $F$ as a future party $\mathsf F$ whose input is $F$ and whose output is one-dimensional. 
Write $\widehat W$ for the same Choi operator viewed in the corresponding $(N+2)$-party process space.

\begin{corollaryone}[Global past and future systems and the $(N+2)$-party representation]
Let $N\ge1$, let $P$ and $F$ be finite-dimensional systems, and let $W$ be a deterministic process. Then the following are equivalent:
$W\in\QCCC_N^{P\to F},
\;
\widehat W\in\QCCC_{N+2},
\;
\widehat W\in\ECS_{N+2}.$
\end{corollaryone}

We first relate the two QC-CC representations by reordering the endpoint parties. The global past party has only an output and can be moved to the first position in every nonzero complete order, with its output stored in an internal ancillary system. The global future party has only an input and can be moved to the last position. This transformation preserves the process operator and the relative order of the intermediate laboratories, proving the first equivalence. 

The ECS equivalence then follows from the main theorem with trivial global past and future. The reordering of the $P,F$ parties to the first and last positions and the conversion between the two representations are proved in SM, Propositions~\ref{prop:pf-qccc-endpoint-reordering} and~\ref{prop:pf-consistency-with-wechs}. Figure~\ref{fig:pf-endpoint-reordering} illustrates the reordering for a complete four-operation order with $s=\mathsf P$. The symbols $\widehat{\mathcal M}$ and $\widetilde{\mathcal M}$ denote the internal CP maps before and after reordering, while $\alpha_i$ and $\eta_i$ denote the corresponding internal ancillary systems. When $P,F$ are both one-dimensional, the two parties become trivial, and the specialization of the main theorem to trivial $P,F$ gives $\ECS_N=\QCCC_N$.

\smallskip
\noindent\textit{Probabilistic processes.\textemdash}
Fix a finite outcome set $\mathcal R$. A family of PSD operators
$\mathbf W=(W^{[r]})_{r\in\mathcal R}$ is an outcome-indexed process family if its sum over outcomes,
$\overline W:=\sum_{r\in\mathcal R}W^{[r]}$, is a deterministic $P\to F$ process matrix.
The family is a probabilistic QC-CC process if it admits a common QC-CC realization, in which the final map applied after all external operations is replaced by a quantum instrument whose outcomes are indexed by $r\in\mathcal R$.
We denote the class of such families by $\pQCCC_N^{P\to F}$ and write $\mathbf W\in\pQCCC_N^{P\to F}$.
Equivalently, each outcome element admits PSD components indexed by complete orders, and the corresponding prefix sums, after summing over outcomes, satisfy Eqs.~
\eqref{eq:qccc-decomposition}--\eqref{eq:qccc-sequence-constraint}.
This extends the QC-CC operator characterization above to the outcome-resolved setting.

For a general PSD outcome-indexed process family $\mathbf X=(X^{[r]})_{r\in\mathcal R}$, the unnormalized probabilistic version of ECS is defined as follows: for arbitrary finite-dimensional input ancillas and any normalized shared state,
the family $(X^{[r]}\otimes\rho)_{r\in\mathcal R}$ admits a common recursive CS first-level decomposition. Each first-level branch is required to belong to
$\mathsf{Proc}_N^{+,P\to F}$ only after summing over outcomes, and the conditional outcome-indexed process family must continue to satisfy the CS condition. 
A family satisfying these conditions is said to belong to $\pECS_N^{+,P\to F}$, and we write $\mathbf X\in\pECS_N^{+,P\to F}$.
If the sum over outcomes $\overline X:=\sum_{r\in\mathcal R}X^{[r]}$ is a deterministic $P\to F$ process matrix, denote the family by $\mathbf W$ and write $\mathbf W\in\pECS_N^{P\to F}$. 
The full proof of the equivalence is given in SM, Proposition~\ref{prop:pecs-pqccc-equivalence}. See Wechs et al.~\cite{wechs2021QuantumCircuitsClassical} for the characterization of probabilistic QC-CC.

\begin{figure}[t]
\centering
\tikzset{
  qcinternal/.style={draw=red!72!black,fill=red!3,line width=.55pt,
    minimum width=1.42cm,minimum height=1.12cm,inner xsep=2pt,inner ysep=2pt,
    align=center,font=\scriptsize},
  qcexternal/.style={draw=blue!65!black,fill=blue!12,line width=.5pt,
    minimum width=.62cm,minimum height=.76cm,inner sep=1pt,font=\scriptsize}
}
\resizebox{0.45\textwidth}{!}{%
\begin{quantikz}[wire types={q,n},column sep=.33cm,row sep=.20cm,
  background color=red!8,thin lines]
\lstick{\textbf{original}}\setwiretype{n} &
 \gate[2,style={qcinternal}]{\widehat{\mathcal M}_{\varnothing}^{\to A}} &
 \gate[style={qcexternal}]{A}\wire[l][1]["A_I^A"{above,pos=.5,yshift=14pt,font=\tiny}]{q} &
 \gate[2,style={qcinternal}]{\widehat{\mathcal M}_{(A)}^{\to s}}\wire[l][1]["A_O^A"{above,pos=.5,yshift=14pt,font=\tiny}]{q} &
 \gate[style={qcexternal}]{s}\setwiretype{n} &
 \gate[2,style={qcinternal}]{\widehat{\mathcal M}_{(A,s)}^{\to B}}\setwiretype{q}\wire[l][1]["A_O^s"{above,pos=.5,yshift=14pt,font=\tiny}]{q} &
 \gate[style={qcexternal}]{B}\wire[l][1]["A_I^B"{above,pos=.5,yshift=14pt,font=\tiny}]{q} &
 \gate[2,style={qcinternal}]{\widehat{\mathcal M}_{(A,s,B)}^{\to C}}\wire[l][1]["A_O^B"{above,pos=.5,yshift=14pt,font=\tiny}]{q} &
 \gate[style={qcexternal}]{C}\wire[l][1]["A_I^C"{above,pos=.5,yshift=14pt,font=\tiny}]{q} &
 \gate[2,style={qcinternal}]{\widehat{\mathcal M}_{(A,s,B,C)}^{\rm fin}}\setwiretype{q}\wire[r][1][""{above}]{n} \\
 & & & \wire[l][2]["\alpha_1"{below,pos=.5,yshift=-3pt,text=red!70!black,font=\scriptsize}]{q}
 & & \wire[l][2]["\alpha_2"{below,pos=.5,yshift=-3pt,text=red!70!black,font=\scriptsize}]{q}
 & & \wire[l][2]["\alpha_3"{below,pos=.5,yshift=-3pt,text=red!70!black,font=\scriptsize}]{q}
 & & \wire[l][2]["\alpha_4"{below,pos=.5,yshift=-3pt,text=red!70!black,font=\scriptsize}]{q}
\end{quantikz}}

\vspace{14mm}
\resizebox{0.45\textwidth}{!}{%
\begin{quantikz}[wire types={q,n},column sep=.33cm,row sep=.20cm,
  background color=red!8,thin lines]
\lstick{\textbf{new}}\setwiretype{n} &
 \gate[style={qcinternal}]{\widetilde{\mathcal M}_{\varnothing}^{\to s}=1} &
 \gate[style={qcexternal}]{s}\setwiretype{n} &
 \gate[2,style={qcinternal}]{\widetilde{\mathcal M}_{(s)}^{\to A}}\setwiretype{q}\wire[l][1]["A_O^s"{above,pos=.5,yshift=14pt,font=\tiny}]{q} &
 \gate[style={qcexternal}]{A}\wire[l][1]["A_I^A"{above,pos=.5,yshift=14pt,font=\tiny}]{q} &
 \gate[2,style={qcinternal}]{\widetilde{\mathcal M}_{(s,A)}^{\to B}}\wire[l][1]["A_O^A"{above,pos=.5,yshift=14pt,font=\tiny}]{q} &
 \gate[style={qcexternal}]{B}\wire[l][1]["A_I^B"{above,pos=.5,yshift=14pt,font=\tiny}]{q} &
 \gate[2,style={qcinternal}]{\widetilde{\mathcal M}_{(s,A,B)}^{\to C}}\wire[l][1]["A_O^B"{above,pos=.5,yshift=14pt,font=\tiny}]{q} &
 \gate[style={qcexternal}]{C}\wire[l][1]["A_I^C"{above,pos=.5,yshift=14pt,font=\tiny}]{q} &
 \gate[2,style={qcinternal}]{\widetilde{\mathcal M}_{(s,A,B,C)}^{\rm fin}}\setwiretype{q}\wire[r][1][""{above}]{n} \\
 & & & & & \wire[l][2]["\eta_2"{below,pos=.5,yshift=-3pt,text=red!70!black,font=\scriptsize}]{q}
 & & \wire[l][2]["\eta_3"{below,pos=.5,yshift=-3pt,text=red!70!black,font=\scriptsize}]{q}
 & & \wire[l][2]["\eta_4"{below,pos=.5,yshift=-3pt,text=red!70!black,font=\scriptsize}]{q}
\end{quantikz}}
\caption{QC-CC circuit diagram for reordering the $P$ parties to the first positions.
 $A,B,C$ label three ordinary laboratories, and $s=\mathsf P$. The upper row shows the original complete order of operations $(A,s,B,C)$, and the lower row shows the order in which the operations are applied in the new circuit, $(s,A,B,C)$.}
\label{fig:pf-endpoint-reordering}
\end{figure}
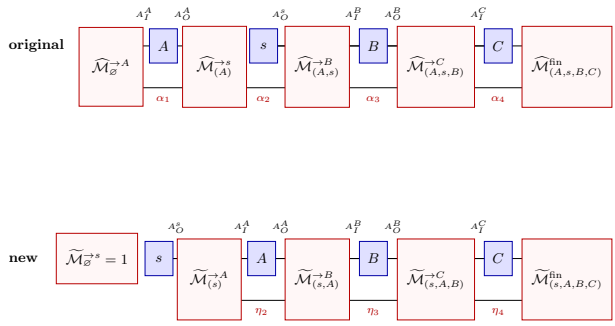

\smallskip\noindent\textit{Conclusion.\textemdash}We prove that ECS processes are exactly those realizable as QC-CC. This equivalence yields a finite semidefinite-programming characterization and an explicit membership test for ECS.
Consequently, existing QC-CC results also apply to ECS, including recent analyses of dynamical causal order and quantum metrology~\cite{mothe2025CorrelationsQuantumCircuits,mothe2024ReassessingAdvantageIndefinite}.
Conversely, the characterization makes techniques developed for ECS available to QC-CC and may facilitate the extension of semi-device-independent certification with trusted quantum inputs to general multipartite processes~\cite{dourdent2022SemiDeviceIndependentCertificationCausal}.

\textit{Acknowledgments.\textemdash}This work is supported by the National Natural Science Foundation of China (Grant No. 12075323), the Natural Science Foundation of Guangdong Province of China (Grant No. 2025A1515011440), and the Innovation Program for Quantum Science and Technology (Grant No. 2021ZD0300702).

\bibliographystyle{apsrev4-2}
\bibliography{references}

\clearpage
\onecolumngrid
\setcounter{page}{1}
\setcounter{section}{0}
\setcounter{subsection}{0}
\setcounter{subsubsection}{0}
\setcounter{equation}{0}
\setcounter{figure}{0}
\setcounter{table}{0}
\setcounter{theorem}{0}
\numberwithin{theorem}{section}
\setlength{\parskip}{0pt}
\setlength{\parindent}{1.5em}
\setlength{\emergencystretch}{2em}
\setlength{\abovedisplayskip}{4pt plus 1pt minus 2pt}
\setlength{\belowdisplayskip}{4pt plus 1pt minus 2pt}
\setlength{\abovedisplayshortskip}{2pt plus 1pt}
\setlength{\belowdisplayshortskip}{2pt plus 1pt}
\setlist{nosep,leftmargin=2em}

\setcounter{secnumdepth}{3}
\renewcommand{\thesection}{S\arabic{section}}
\renewcommand{\thesubsection}{\thesection.\Alph{subsection}}
\renewcommand{\thesubsubsection}{\thesubsection.\arabic{subsubsection}}
\renewcommand{\theequation}{S\arabic{equation}}
\makeatletter
\def\p@subsection{}
\def\p@subsubsection{}
\def\p@paragraph{}
\def\p@subparagraph{}
\makeatother

\begin{center}
{\Large\bfseries Supplemental Material for\\[0.3em]
``Extensibly Causally Separable Processes Admit Realizations as Quantum Circuits with Classical Control of Causal Order''\par}
\vspace{1em}
{\large Wenjie Wei$^{1}$\orcidlink{0009-0001-2435-0852} and
Shengshi Pang$^{1,2}$\orcidlink{0000-0002-6351-539X}\par}
\vspace{0.5em}
{\small $^{1}$School of Physics, Sun Yat-sen University, Guangzhou, Guangdong 510275, China\par
$^{2}$Hefei National Laboratory, University of Science and Technology of China, Hefei 230088, China\par}
\vspace{0.4em}
{\small\href{mailto:pangshsh@mail.sysu.edu.cn}{\texttt{pangshsh@mail.sysu.edu.cn}}\par}
\end{center}
\vspace{1.2em}

\noindent
This Supplemental Material presents the definitions used in the main text and a complete proof of the main theorem. We abbreviate extensible causal separability as ECS and quantum circuits with classical control of causal order as QC-CC.
Section~\ref{sec:foundations} defines the general $P\to F$ framework. Section~\ref{sec:pf} proves the main theorem. Section~\ref{sec:pf-consistency} relates it to the ECS definition with trivial global past and future through the $(N+2)$-party representation, and Sec.~\ref{chap:probabilistic-ecs} treats the probabilistic extension.

\section{Basic definitions}

\label{sec:foundations}

Let $\N=\{1,\ldots,N\}$ label the $N$ external operations. The input and output Hilbert spaces associated with operation $k$ are denoted by $A_I^k,A_O^k$, respectively. We write
$A_{IO}^k:=A_I^k\otimes A_O^k$ and, for any finite label set $\mathcal S\subseteq\N$, adopt the conventions
$A_{IO}^{\mathcal S}:=\bigotimes_{k\in\mathcal S}A_{IO}^k$, $A_{IO}^{\varnothing}:=\mathbb C$, and
$d_{A_{IO}^{\mathcal S}}:=\prod_{k\in\mathcal S}d_{A_I^k}d_{A_O^k}$. Juxtaposed system symbols denote the corresponding tensor product.
Label sets are written explicitly as $\N$. Representing $P,F$ by two parties with one-dimensional input and output ports, respectively, gives a total of $N+2$ parties.
All spaces are finite dimensional. We use $d_X$ for the dimension of $X$, $\id^X$ for the identity operator on $X$, and $\Pos(X)$ for the cone of positive semidefinite (PSD) operators on $X$. For any Hilbert space $H$, $\mathcal L(H)$ denotes the space of linear operators on $H$. For any finite label set $\mathcal S$, define $D_O(\mathcal S):=\prod_{k\in\mathcal S}d_{A_O^k},
   \; D_O(\varnothing):=1.$

If party $k$ has an input ancilla $E_k$, fix a quantum operation whose Choi--Jamiołkowski (CJ) matrix acts on $A_I^k\otimes E_k\otimes A_O^k$ and is denoted by $C_k$, so that
$C_k\succeq0$ and $\Tr_{A_O^k}C_k\preceq\id^{A_I^kE_k}$. For any extended operator $X$, the conditioning of $X$ on $C_k$ is defined by
\begin{equation}
 X_{|C_k}:=\Tr_{A_I^kE_kA_O^k}\!\left[(C_k\otimes\id)X\right].
\label{eq:basic-condition}
\end{equation}

For a nonempty ordered sequence $h=(k_1,\ldots,k_s)$ and a tuple $\boldsymbol C_h=(C_{k_1},\ldots,C_{k_s})$ of quantum-operation Choi matrices fixed along that history, define
\begin{equation}
X_{|\boldsymbol C_h}:=\bigl(\cdots((X_{|C_{k_1}})_{|C_{k_2}})\cdots\bigr)_{|C_{k_s}},
\;
\boldsymbol C_{\varnothing}:=(),\; X_{|\boldsymbol C_{\varnothing}}:=X.
\label{eq:basic-iterated-condition}
\end{equation}

For two operators sharing a system $H$, let $T_H$ denote the partial transpose on $H$ in the fixed CJ basis, and define the link product by
\begin{equation}
X*Y:=\Tr_H\!\left[(X^{T_H}\otimes\id)Y\right].
\label{eq:basic-link-product}
\end{equation}
With the convention in Eq.~\eqref{eq:basic-link-product},
$X_{|C_k}=C_k^T*X$. Since full transposition bijectively preserves the quantum-operation CJ cone, it does not affect any universal condition below.

For any subsystem $R$, define the completely depolarizing superoperator by
\begin{equation}
 \mathcal D_R(Y):=\Tr_R(Y)\otimes\frac{\id^R}{d_R},
 \; \mathcal I_R(Y):=Y .
\label{eq:basic-replace}
\end{equation}
The system subscript of the identity superoperator is omitted when its domain is clear.
If $R$ and $S$ are disjoint, then
$\mathcal D_R\mathcal D_S=\mathcal D_S\mathcal D_R=\mathcal D_{RS}$.

For a finite label set $\mathcal T$, $\Pi_{\mathcal T}$ denotes the set of its full permutations. A history is an ordered sequence of distinct labels in $\mathcal T$. For a
nonempty history $h=(k_1,\ldots,k_n)$, write $\ell(h):=k_n$. The notation $(h,j)$ or $(k,h)$ denotes the flat concatenation of one label to the end or beginning of $h$, respectively. We also adopt the convention that
$\pi\succeq\varnothing$ holds for all $\pi\in\Pi_{\mathcal T}$, whereas, for nonempty $h$,
$\pi\succeq h$ if and only if $(\pi_1,\ldots,\pi_n)=h$.

The conditions defining a deterministic $P\to F$ process matrix comprise three parts: positive semidefiniteness, homogeneous linear equality constraints, and an affine normalization constraint ensuring that every insertion of local CPTP maps yields a trace-preserving $P\to F$ map. The cone of unnormalized processes used here is obtained by dropping the affine constraint and retaining the first two sets of conditions. 
Ordinary multipartite process matrices correspond to $P\simeq F\simeq\mathbb C$~\cite{araujo2015WitnessingCausalNonseparability}.

\begin{proposition}[Equivalent characterizations of the cone of unnormalized $P\to F$ processes]
\label{prop:pf-validity}
Let $N\ge1$ and
$X\in\Pos\!\left(P\otimes A_{IO}^{\N}\otimes F\right),\; \Theta_X:=\frac{1}{D_O(\N)}\Tr_{A_{IO}^{\N}F}X\in\Pos(P).$
The following conditions are equivalent:
\begin{enumerate}[label=(\roman*),leftmargin=2.4em]
\item For arbitrary CJ matrices $M_i$ of local completely positive and trace-preserving (CPTP) maps,
\begin{equation}
\Tr_F\!\left[\left(\bigotimes_{i\in\N}M_i\right)*X\right]=\Theta_X,
\end{equation}
that is, taking the resulting $P\to F$ CJ matrix and tracing out its output system $F$ always gives $\Theta_X$, independently of the choice of local CPTP maps.

\item For every nonempty subset $\varnothing\ne\mathcal K\subseteq\N$,
\begin{equation}
\left(\prod_{i\in\mathcal K}(\mathcal I-\mathcal D_{A_O^i})\right)
\mathcal D_{A_{IO}^{\N\setminus\mathcal K}F}X=0.
\label{eq:basic-process-validity}
\end{equation}
\end{enumerate}
The positive operators satisfying either condition form the cone of unnormalized valid processes, denoted by $\mathsf{Proc}_N^{+,P\to F}$, with the convention $\mathsf{Proc}_0^{+,P\to F}:=\Pos(P\otimes F)$ and $\Theta_X:=\Tr_F X$ for $N=0$. Furthermore, $X$ is a deterministic $P\to F$ process matrix if and only if these equivalent conditions hold and
\begin{equation}
\Theta_X=\id^P
\;\;\text{if and only if}\;\;
\Tr_{A_{IO}^{\N}F}X=D_O(\N)\,\id^P.
\label{eq:pf-initial-normalization}
\end{equation}
The homogeneous constraints in Eq.~\eqref{eq:basic-process-validity} require the output trace to be independent of all inserted CPTP maps. The resulting positive operator $\Theta_X$ is fixed by $X$. Equation~\eqref{eq:pf-initial-normalization} selects the deterministic section by imposing $\Theta_X=\id^P$. When $P$ is one dimensional, the latter reduces to $\Tr X=D_O(\N)$. 
\end{proposition}

\begin{proof}
For each $i\in\N$, choose the full-rank CPTP CJ matrix
$D_i:=\id^{A_I^iA_O^i}/d_{A_O^i}$. With $D_i$ as a reference point, the translation directions of the affine space of Hermitian trace-preserving (TP) CJ operators satisfy
$\Tr_{A_O^i}H_i=0$, and this direction space is precisely $\operatorname{Im}(\mathcal I-\mathcal D_{A_O^i})$. Writing the CPTP map inserted in each slot as
$D_i+t_iH_i$ for sufficiently small real $t_i$ and taking the multiaffine expansion, condition (i) is independent of the inserted maps if and only if every nonconstant coefficient indexed by a nonempty set $\mathcal K$ vanishes. Nondegeneracy of the finite-dimensional
Hilbert--Schmidt pairing expresses these coefficient conditions as
Eq.~\eqref{eq:basic-process-validity}. Conversely, the same expansion and Eq.~\eqref{eq:basic-process-validity}
eliminate all nonconstant terms, making the result independent of the inserted CPTP maps. Setting every slot to $D_i$ gives the remaining constant term
\[
\Tr_F\!\left[\left(\bigotimes_{i\in\N}D_i\right)*X\right]
=\frac{1}{D_O(\N)}\Tr_{A_{IO}^{\N}F}X
=\Theta_X,
\]
which establishes (i). If, in addition,
$\Theta_X=\id^P$, every insertion of CPTP maps yields a trace-preserving $P\to F$ CJ matrix. 
Since $X\succeq0$ guarantees
complete positivity, this trace-preservation property makes $X$ a valid $P\to F$ process matrix.
\end{proof}

When $P\simeq F\simeq\mathbb C$, Eq.~\eqref{eq:basic-process-validity} reduces to the linear conditions for ordinary multipartite process matrices used by Wechs et al.~\cite{wechs2019DefinitionCharacterisationMultipartite}. We denote the corresponding linear subspace of valid processes by $\mathcal L^{\N}$. Deterministic normalization additionally requires $\Tr X=D_O(\N)$.
The ordinary cone of unnormalized processes therefore satisfies
\begin{equation}
\mathsf{Proc}_N^{+,\mathbb C\to\mathbb C}
=\{X\succeq0:X\in\mathcal L^{\N}\}.
\label{eq:pf-proc-flat}
\end{equation}

If party $j$ has an input ancilla $E_j$, its full input-output system is $A_I^jE_jA_O^j$.

Let $\mathcal L^{P\prec k\prec(\N\setminus\{k\})\prec F}$ denote the subspace defined by
\begin{align}
   (\mathcal I-\mathcal D_{A_O^k})
   \mathcal D_{A_{IO}^{\N\setminus\{k\}}F}Y&=0,
   \label{eq:pf-firstcompat}\\
   \left(\prod_{i\in\mathcal X}
      (\mathcal I-\mathcal D_{A_O^i})\right)
   \mathcal D_{A_{IO}^{(\N\setminus\{k\})\setminus\mathcal X}F}Y&=0
   \;
   (\varnothing\ne\mathcal X\subseteq\N\setminus\{k\})
   \label{eq:pf-auto-first-rest}
\end{align}
for $Y\in\Herm(P\otimes A_{IO}^{\N}\otimes F)$. If $Y\in\mathsf{Proc}_N^{+,P\to F}$ also belongs to this subspace, it is called compatible with $P\prec k\prec(\N\setminus\{k\})\prec F$.

We then define the CS and ECS processes for processes with nontrivial $P$ and $F$ as follows.

\begin{definition}[ECS for processes with nontrivial $PF$]
\label{def:pf-direct-ecs}
The base case of the recursion is
\begin{equation}
   \CS_0^{+,P\to F}:=\Pos(P\otimes F)
   =\mathsf{Proc}_0^{+,P\to F}.
   \label{eq:pf-CS-base}
\end{equation}
For $N\ge1$, $X\in\CS_N^{+,P\to F}$ if and only if there exists a decomposition
\begin{equation}
   X=\sum_{k\in\N}X_{(k)},
   \;
   X_{(k)}\in\mathsf{Proc}_N^{+,P\to F},
   \label{eq:pf-CSsplit}
\end{equation}
  such that, for any quantum operation by party $k$ with CJ matrix $C_k$,
\begin{equation}
   (X_{(k)})_{|C_k}\in\CS_{N-1}^{+,P\to F}.
   \label{eq:pf-CSrec}
\end{equation}
If the inputs have been extended to $A_I^jE_j$, Eq.~\eqref{eq:pf-CSsplit} is understood in terms of the parties' input-output systems $A_I^jE_jA_O^j$. The subscript $N-1$ on the right-hand side records only the number of remaining operations, whose actual labels are $\N\setminus\{k\}$. %The same $P,F$ are retained throughout.

Set $\ECS_0^{+,P\to F}:=\Pos(P\otimes F)$. For $N\ge1$, an operator $X\in\mathsf{Proc}_N^{+,P\to F}$ belongs to $\ECS_N^{+,P\to F}$ if and only if, for arbitrary input ancillas state $\rho$, the extended operator $X\otimes\rho$ belongs to $\CS_N^{+,P\to F}$ as defined above. The deterministic section is
\[
   \ECS_N^{P\to F}:=
   \left\{X\in\ECS_N^{+,P\to F}\mid \Theta_X=\id^P\right\}.
\]
For a deterministic process $W=\sum_kX_{(k)}$, the first-level components satisfy $\Theta_{X_{(k)}}\succeq0$ and $\sum_k\Theta_{X_{(k)}}=\id^P$. These operators describe the input dependence of the branch weights.

\end{definition}

When $P\simeq F\simeq\mathbb C$, Definition~\ref{def:pf-direct-ecs} reduces to the ordinary multipartite definitions of CS and ECS~\cite{ognyanoreshkov2016CausalCausallySeparable,wechs2019DefinitionCharacterisationMultipartite}. The following definition records the notation and recursion for this case.

\label{sec:basic-cs-ecs}

\begin{definition}[The causally separable cone and extensible causal separability]
\label{def:basic-cs-ecs}
The unnormalized causally separable (CS) cone is defined recursively as follows. Set
$\CSplus_0:=\mathbb R_{\ge0}.$
For $N=1$, $\CSplus_1$ is the cone of unnormalized valid single-party processes. For $N\ge2$,
$X\in\CSplus_N$ if and only if there exists a PSD decomposition
\begin{equation}
   X=\sum_{k\in\N}X_{(k)},
   \; X_{(k)}\succeq0,
   \label{eq:basic-ogcs-split}
\end{equation}
such that every $X_{(k)}$ is an ordinary valid process compatible with $A_k$ acting first (the marginal statistics of $A_k$ are independent of the other parties' instrument choices) and, for any quantum operation by party $k$ with CJ matrix
$C_k$,
\begin{equation}
   (X_{(k)})_{|C_k}\in\CSplus_{N-1}.
   \label{eq:basic-ogcs-rec}
\end{equation}
The section of this cone consisting of normalized valid processes is denoted by $\CS_N$. Equivalent formulations in terms of decompositions and conditioned processes are given by Oreshkov--Giarmatzi and Wechs et al.~\cite{ognyanoreshkov2016CausalCausallySeparable,wechs2019DefinitionCharacterisationMultipartite}.

A normalized process matrix $W$ is called \emph{extensibly causally separable}, and the corresponding property is called extensible causal separability
(ECS), if, for arbitrary input ancillas
$E_1,\ldots,E_N$ and every normalized shared state
$\rho\in\Pos(E_1\otimes\cdots\otimes E_N)$,
extending the input of party $k$ to $A_I^kE_k$ gives
$W\otimes\rho\in\CS_N.$
The normalized set is denoted by $\ECS_N$. The unnormalized notation specializes Definition~\ref{def:pf-direct-ecs} to one-dimensional global systems:
\begin{equation}
   \ECSplus_N:=\ECS_N^{+,\mathbb C\to\mathbb C}.
   \label{eq:basic-ecs-cone}
\end{equation}
See Oreshkov--Giarmatzi and Wechs et al.~\cite{ognyanoreshkov2016CausalCausallySeparable,wechs2019DefinitionCharacterisationMultipartite} for the related definitions.
\end{definition}

Proposition~\ref{lem:pf-uniform} proves the existence of an extension-independent first-level decomposition for general $P,F$. Taking
$P\simeq F\simeq\mathbb C$ recovers the corresponding result for ordinary multipartite ECS~\cite{wechs2019DefinitionCharacterisationMultipartite}.

\label{sec:basic-qccc}

\begin{definition}[QC-CC]
\label{def:basic-qccc}

Fix a finite-dimensional global past $P$, a global future $F$, and
external operations with input and output spaces $A_I^k$ and $A_O^k$,
respectively, for $k\in\N$.

A deterministic $P\to F$ process matrix $W$ belongs to
$\QCCC_N^{P\to F}$ if and only if there exist finite-dimensional
internal ancillary systems and the following internal operations.
At the initial stage, there is a quantum instrument
$\{\mathcal M_{\varnothing}^{\to k}\}_{k\in\N}$.
After every incomplete history $h=(k_1,\ldots,k_n)$, there is a
quantum instrument
$\{\mathcal M_h^{\to j}\}_{j\in
\N\setminus\{k_1,\ldots,k_n\}}$.
After every full permutation $\pi\in\Pi_{\N}$, there is a final map
$\mathcal M_\pi^{\rm fin}$. The initial instrument takes the global input $P$ to $A_I^k$ together
with its internal memory. The instrument after $h$ takes
$A_O^{\ell(h)}$ together with the current memory to $A_I^j$ together
with the next memory. The final map takes $A_O^{\pi_N}$ together with
the remaining memory to the global future $F$.
For suitable finite-dimensional memories
$\alpha_1,\ldots,\alpha_N$, let
$M_{\varnothing}^{\to k}\in\Pos(PA_I^k\alpha_1)$,
$M_h^{\to j}\in
\Pos(A_O^{\ell(h)}\alpha_{|h|}A_I^j\alpha_{|h|+1})$, and
$M_\pi^{\rm fin}\in\Pos(A_O^{\pi_N}\alpha_N F)$
denote the corresponding CJ operators.

The trace-preservation conditions are
\begin{equation}
\sum_{k\in\N}
\Tr_{A_I^k\alpha_1}M_{\varnothing}^{\to k}
=
\id^P.
\label{eq:basic-qccc-tp-initial}
\end{equation}
For every incomplete history $h=(k_1,\ldots,k_n)$ with $1\le n<N$,
\begin{equation}
\begin{aligned}
&\sum_{\substack{j\in\N\\
j\notin\{k_1,\ldots,k_n\}}}
\Tr_{A_I^j\alpha_{h,j}}\!\Bigl[
M_{\varnothing}^{\to k_1}*
M_{(k_1)}^{\to k_2}*\cdots*
M_{(k_1,\ldots,k_{n-1})}^{\to k_n}*
M_h^{\to j}
\Bigr]
\\
&\quad =
\Tr_{\alpha_h}\!\Bigl[
M_{\varnothing}^{\to k_1}*
M_{(k_1)}^{\to k_2}*\cdots*
M_{(k_1,\ldots,k_{n-1})}^{\to k_n}
\Bigr]
\otimes\id^{A_O^{\ell(h)}}.
\end{aligned}
\label{eq:basic-qccc-tp-intermediate}
\end{equation}
For every full permutation
$\pi=(\pi_1,\ldots,\pi_N)\in\Pi_{\N}$,
\begin{equation}
\begin{aligned}
&\Tr_F\!\Bigl[
M_{\varnothing}^{\to\pi_1}*
M_{(\pi_1)}^{\to\pi_2}*\cdots*
M_{(\pi_1,\ldots,\pi_{N-1})}^{\to\pi_N}*
M_\pi^{\rm fin}
\Bigr]
\\
&\quad =
\Tr_{\alpha_\pi}\!\Bigl[
M_{\varnothing}^{\to\pi_1}*
M_{(\pi_1)}^{\to\pi_2}*\cdots*
M_{(\pi_1,\ldots,\pi_{N-1})}^{\to\pi_N}
\Bigr]
\otimes\id^{A_O^{\pi_N}}.
\end{aligned}
\label{eq:basic-qccc-tp-final}
\end{equation}
For $n=1$, the products in
Eq.~\eqref{eq:basic-qccc-tp-intermediate} contain no intermediate
factor between $M_{\varnothing}^{\to k_1}$ and $M_h^{\to j}$.

The resulting process matrix is
\begin{equation}
W=
\sum_{\pi\in\Pi_{\N}}
M_\pi^{\rm fin}*
M_{(\pi_1,\ldots,\pi_{N-1})}^{\to\pi_N}*
\cdots*
M_{(\pi_1)}^{\to\pi_2}*
M_{\varnothing}^{\to\pi_1}.
\label{eq:basic-qccc-operational}
\end{equation}

For $N=1$, there are no intermediate maps or intermediate
normalization conditions, and the expression contains only
$M_{(k)}^{\rm fin}*M_{\varnothing}^{\to k}$.
For $N=0$, $\QCCC_0^{P\to F}$ is defined as the set of CJ operators
of all CPTP maps from $P$ to $F$. See Wechs et al.~\cite{wechs2021QuantumCircuitsClassical} for details.

\end{definition}
\begin{proposition}[Characterization of QC-CC]
\label{prop:basic-pf-qccc}
Let $N\ge1$, and let $W$ be a deterministic $P\to F$ process matrix. The following three statements are equivalent:
\begin{enumerate}[label=(\roman*),leftmargin=2.4em]
\item $W\in\QCCC_N^{P\to F}$.

\item For every nonempty ordered sequence $h=(k_1,\ldots,k_n)$, there exists a PSD reduced operator\linebreak
$Q_h\in \Pos\!\left( P\otimes A_{IO}^{\{k_1,\ldots,k_{n-1}\}}\otimes A_I^{k_n} \right),$
and, for every full permutation $\pi\in\Pi_{\N}$, there exists
$T_\pi\in\Pos(P\otimes A_{IO}^{\N}\otimes F),$
such that
\begin{align}
   W&=\sum_{\pi\in\Pi_{\N}}T_\pi,
   \label{eq:basic-pf-qccc-sum}\\
   \sum_{k\in\N}\Tr_{A_I^k}Q_{(k)}&=\id^P,
   \label{eq:basic-pf-qccc-initial}\\
   \sum_{j\in\N\setminus\{k_1,\ldots,k_n\}}\Tr_{A_I^j}Q_{(h,j)}
   &=Q_h\otimes\id^{A_O^{\ell(h)}}
   \;(|h|<N),
   \label{eq:basic-pf-qccc-rec}\\
   \Tr_FT_\pi&=Q_\pi\otimes\id^{A_O^{\ell(\pi)}}.
   \label{eq:basic-pf-qccc-term}
\end{align}

\item There exist PSD operators indexed by full permutations $\pi$,
$\{T_\pi\}_{\pi\in\Pi_{\N}}$, and, for every nonempty ordered sequence $h$, the partial sum specified by $h$ is defined as
\begin{equation}
   T_h:=\sum_{\pi\succeq h}T_\pi.
   \label{eq:basic-pf-Th}
\end{equation}
Here the sum includes all extensions of $h$ to full permutations $\pi$. These operators satisfy
\begin{align}
  & W=\sum_{\pi\in\Pi_{\N}}T_\pi,
   \; T_\pi\succeq0,
   \label{eq:basic-pf-permutation-sum}\\
  & (\mathcal I-\mathcal D_{A_O^{\ell(h)}})
   \mathcal D_{A_{IO}^{\N\setminus\{k_1,\ldots,k_n\}}F}T_h=0
   \;(\varnothing\ne h).
   \label{eq:basic-pf-sequence-constraint}
\end{align}
\end{enumerate}
The reduced-history formulation (ii) and the full-space prefix-sum formulation (iii) encode the same QC-CC constraints. Their correspondence uniquely determines the reduced operators from the history partial sums as
\begin{equation}
   Q_h=
   \frac{1}{d_{A_O^{\ell(h)}}D_O(\N\setminus\{k_1,\ldots,k_n\})}
   \Tr_{A_O^{\ell(h)}A_{IO}^{\N\setminus\{k_1,\ldots,k_n\}}F}T_h .
   \label{eq:basic-pf-Qdictionary}
\end{equation}
\end{proposition}

\begin{proof}
The equivalence of (i) and (ii) follows from Proposition 5 and Appendix B.2.c of Wechs et al.~\cite{wechs2021QuantumCircuitsClassical}, which include global input and output systems $P,F$. Their trace-preservation conditions for the initial, intermediate, and final internal instruments are Eqs.~\eqref{eq:basic-pf-qccc-initial}--\eqref{eq:basic-pf-qccc-term} in our notation. The realization construction in their Eq.~(B31) preserves the specified reduced histories $Q_h$ and complete-history contributions $T_\pi$.

We next prove (ii)$\Longleftrightarrow$(iii), relating the reduced history operators $Q_h$ to the full-permutation components $T_\pi$. First assume (iii). By
Eq.~\eqref{eq:basic-pf-sequence-constraint},
taking $T_h$ and tracing out the uncalled slots and $F$ gives a reduced operator that factors as a tensor product with the identity on $A_O^{\ell(h)}$. Thus,
Eq.~\eqref{eq:basic-pf-Qdictionary} gives
\begin{equation}
   \frac{1}{D_O(\N\setminus\{k_1,\ldots,k_n\})}
   \Tr_{A_{IO}^{\N\setminus\{k_1,\ldots,k_n\}}F}T_h
   =Q_h\otimes\id^{A_O^{\ell(h)}}.
   \label{eq:basic-pf-prefix-reduction}
\end{equation}
If $|h|<N$, use
$T_h=\sum_{j\in\N\setminus\{k_1,\ldots,k_n\}}T_{(h,j)}$, apply
Eq.~\eqref{eq:basic-pf-prefix-reduction} to each successor, and trace over $A_I^j$ to obtain
\begin{equation}
   \sum_{j\in\N\setminus\{k_1,\ldots,k_n\}}\Tr_{A_I^j}Q_{(h,j)}
=Q_h\otimes\id^{A_O^{\ell(h)}}.
\end{equation}
At a sequence $h=\pi$ corresponding to a full permutation, Eq.~\eqref{eq:basic-pf-prefix-reduction} gives $\Tr_FT_\pi=Q_\pi\otimes\id^{A_O^{\ell(\pi)}}.$
Finally, the initial normalization of the deterministic process $W$ in
Eq.~\eqref{eq:pf-initial-normalization} gives
\begin{equation}
\sum_{k\in\N}\Tr_{A_I^k}Q_{(k)}
=\frac1{D_O(\N)} \Tr_{A_{IO}^{\N}F}W
=\id^P.
\end{equation}
Thus (ii) holds.

Conversely, assume (ii). We prove Eq.~\eqref{eq:basic-pf-prefix-reduction} by backward induction on the sequence length, from $N$ to $1$.
For sequences corresponding to full permutations, the identity is precisely
Eq.~\eqref{eq:basic-pf-qccc-term}, which establishes the base case. If it holds for all ordered sequences of length $n+1$,
then, for any $|h|=n$,
$T_h=\sum_{j\in\N\setminus\{k_1,\ldots,k_n\}}T_{(h,j)}.$
Applying the induction hypothesis term by term and using
$D_O(\N\setminus\{k_1,\ldots,k_n\})=d_{A_O^j}D_O(\N\setminus\{k_1,\ldots,k_n,j\})$ yields
\begin{align*}
   \frac1{D_O(\N\setminus\{k_1,\ldots,k_n\})}
   \Tr_{A_{IO}^{\N\setminus\{k_1,\ldots,k_n\}}F}T_h
   &=\sum_{j\in\N\setminus\{k_1,\ldots,k_n\}}\Tr_{A_I^j}Q_{(h,j)}\\
   &=Q_h\otimes\id^{A_O^{\ell(h)}}.
\end{align*}
Hence Eq.~\eqref{eq:basic-pf-prefix-reduction} holds for every nonempty ordered sequence. To make its equivalence to
Eq.~\eqref{eq:basic-pf-sequence-constraint} explicit, fix a nonempty ordered sequence $h$ and write
\begin{equation}
   S_h:=\frac1{D_O(\N\setminus\{k_1,\ldots,k_n\})}
   \Tr_{A_{IO}^{\N\setminus\{k_1,\ldots,k_n\}}F}T_h .
   \label{eq:basic-pf-Sh}
\end{equation}
By the definition of the completely depolarizing map, there exists a strictly positive scalar
$\gamma_h:=\frac{D_O(\N\setminus\{k_1,\ldots,k_n\})}{d_{A_{IO}^{\N\setminus\{k_1,\ldots,k_n\}}}d_F}>0,$
such that
\begin{equation}
   \mathcal D_{A_{IO}^{\N\setminus\{k_1,\ldots,k_n\}}F}T_h
   =\gamma_h S_h
      \otimes\id^{A_{IO}^{\N\setminus\{k_1,\ldots,k_n\}}F}.
   \label{eq:basic-pf-sequence-depol-expanded}
\end{equation}
Since $A_O^{\ell(h)}$ and $A_{IO}^{\N\setminus\{k_1,\ldots,k_n\}}F$ are disjoint,
\begin{equation}
 (\mathcal I-\mathcal D_{A_O^{\ell(h)}})
  \mathcal D_{A_{IO}^{\N\setminus\{k_1,\ldots,k_n\}}F}T_h =\gamma_h
  \bigl[(\mathcal I-\mathcal D_{A_O^{\ell(h)}})S_h\bigr]
  \otimes\id^{A_{IO}^{\N\setminus\{k_1,\ldots,k_n\}}F}.
   \label{eq:basic-pf-Sh-invariant}
\end{equation}
Equation~\eqref{eq:basic-pf-prefix-reduction}, established by induction, states precisely that
$S_h=Q_h\otimes\id^{A_O^{\ell(h)}}$, so
$(\mathcal I-\mathcal D_{A_O^{\ell(h)}})S_h =(\mathcal I-\mathcal D_{A_O^{\ell(h)}}) (Q_h\otimes\id^{A_O^{\ell(h)}})=0.$
Substitution into Eq.~\eqref{eq:basic-pf-Sh-invariant} directly yields
Eq.~\eqref{eq:basic-pf-sequence-constraint}.

Conversely, since $\gamma_h>0$ and
$\id^{A_{IO}^{\N\setminus\{k_1,\ldots,k_n\}}F}\ne0$, Eq.~\eqref{eq:basic-pf-sequence-constraint}
is equivalent to the fixed-point condition
\(\mathcal D_{A_O^{\ell(h)}}S_h=S_h\). This condition is in turn equivalent to
\begin{equation}
   S_h=\widetilde Q_h\otimes\id^{A_O^{\ell(h)}}
    \;\text{and}\;
   \widetilde Q_h
   =\frac1{d_{A_O^{\ell(h)}}}\Tr_{A_O^{\ell(h)}}S_h .
   \label{eq:basic-pf-Sh-dictionary}
\end{equation}
Substituting the definition of $S_h$ into Eq.~\eqref{eq:basic-pf-Sh-dictionary} gives
Eq.~\eqref{eq:basic-pf-Qdictionary}. Defining
$Q_h:=\widetilde Q_h$ by this equation recovers Eq.~\eqref{eq:basic-pf-prefix-reduction}. Thus,
Eq.~\eqref{eq:basic-pf-prefix-reduction} and
Eq.~\eqref{eq:basic-pf-sequence-constraint} are equivalent for each nonempty ordered sequence, establishing (iii).
\end{proof}

When $d_P=d_F=1$, Definition~\ref{def:basic-qccc} reduces to ordinary multipartite QC-CC.
The right-hand side of the initial condition in Proposition~\ref{prop:basic-pf-qccc} reduces to the scalar $1$, and the
$F$ factor disappears from every future block. This gives the characterization of ordinary multipartite QC-CC by operators indexed by full permutations.

\section{Proof of the main theorem}
\label{sec:pf}
For a general finite-dimensional global past $P$ and global future $F$, we prove that ECS is equivalent to $\QCCC_N^{P\to F}$. The ordinary multipartite case follows by setting $d_P=d_F=1$.

\begin{theorem}[Equivalence of ECS and QC-CC]
\label{thm:pf-ECS-QCCC}
Let $N\ge1$, and let $W$ be a finite-dimensional deterministic $P\to F$ process matrix. The following conditions are equivalent:
\begin{enumerate}[label=(\roman*),leftmargin=2.4em]
\item $W\in\ECS_N^{P\to F}$.
\item $W\in\QCCC_N^{P\to F}$, that is, there exist PSD operators $\{T_\pi\}_{\pi\in\Pi_{\N}}$ satisfying
\begin{align}
  & W=\sum_{\pi\in\Pi_{\N}}T_\pi,
   \label{eq:pf-main-permutation-sum}\\
  & (\mathcal I-\mathcal D_{A_O^{\ell(h)}})
   \mathcal D_{A_{IO}^{\N\setminus\{k_1,\ldots,k_n\}}F}
   \left(\sum_{\pi\succeq h}T_\pi\right)=0
   \;(\varnothing\ne h).
   \label{eq:pf-main-prefix}
\end{align}
\end{enumerate}
\end{theorem}

\begin{proposition}[Compatibility of first-level components with a specified $P\prec k\prec\cdots\prec F$ order]
\label{prop:pf-first-component-order}
Let $N\ge1$ and
$X\in\CS_N^{+,P\to F}$. Choose any witnessing decomposition from Definition~\ref{def:pf-direct-ecs},
\begin{equation}
   X=\sum_{k\in\N}X_{(k)},
   \;
   X_{(k)}\in\mathsf{Proc}_N^{+,P\to F}.
   \label{eq:pf-sanity-witness}
\end{equation}
Then, for every $k\in\N$,
\begin{equation}
   X_{(k)}
   \in
   \mathcal L^{P\prec k\prec(\N\setminus\{k\})\prec F}.
   \label{eq:pf-sanity-kfirst}
\end{equation}
That is, the first-level component $X_{(k)}$ is compatible with the specified order
$P\prec k\prec(\N\setminus\{k\})\prec F$.

\end{proposition}

\begin{proof}
For $X_{(k)}\in\mathsf{Proc}_N^{+,P\to F}$, the constraint in Eq.~\eqref{eq:basic-process-validity} with $\mathcal K=\{k\}$ reads $(\mathcal I-\mathcal D_{A_O^k}) \mathcal D_{A_{IO}^{\N\setminus\{k\}}F}X_{(k)}=0,$
which is Eq.~\eqref{eq:pf-firstcompat}.

Next take any nonempty
$\mathcal X\subseteq\N\setminus\{k\}$. The recursive condition
in Eq.~\eqref{eq:pf-CSrec} implies that,
for every quantum operation by party $k$ with CJ matrix $C_k$,
$(X_{(k)})_{|C_k} \in\CS_{N-1}^{+,P\to F}.$
By Definition~\ref{def:pf-direct-ecs} and
the convex-cone property of $\mathsf{Proc}_{N-1}^{+,P\to F}$,
$\CS_{N-1}^{+,P\to F} \subseteq \mathsf{Proc}_{N-1}^{+,P\to F}.$
Therefore,
\[
\left(\prod_{i\in\mathcal X}
   (\mathcal I-\mathcal D_{A_O^i})\right)
\mathcal D_{A_{IO}^{(\N\setminus\{k\})\setminus\mathcal X}F}
\bigl((X_{(k)})_{|C_k}\bigr)=0
\;\forall C_k.
\]
The linear map above acts only on the remaining slots and $F$, and therefore commutes with conditioning on party $k$, giving
\[
\left[
\left(\prod_{i\in\mathcal X}
   (\mathcal I-\mathcal D_{A_O^i})\right)
\mathcal D_{A_{IO}^{(\N\setminus\{k\})\setminus\mathcal X}F}
X_{(k)}
\right]_{|C_k}=0
\;\forall C_k.
\]
If a Hermitian operator vanishes upon conditioning on the CJ matrix of every quantum operation by party $k$,
then the operator itself is zero: every rank-one positive operator becomes a
valid CJ matrix of a quantum operation after multiplication by a sufficiently small positive scalar, and rank-one positive operators span the full real vector space of Hermitian operators for party $k$.
Hence
\[
\left(\prod_{i\in\mathcal X}
   (\mathcal I-\mathcal D_{A_O^i})\right)
\mathcal D_{A_{IO}^{(\N\setminus\{k\})\setminus\mathcal X}F}
X_{(k)}=0
\]
for every nonempty $\mathcal X\subseteq\N\setminus\{k\}$, which is
Eq.~\eqref{eq:pf-auto-first-rest}. Together with Eq.~\eqref{eq:pf-firstcompat} and the group-order notation above for fixed $P,F$, this yields Eq.~\eqref{eq:pf-sanity-kfirst}.
\end{proof}

\begin{proposition}[Equivalence of validity of the total operator and validity of the first-level components]
\label{prop:pf-recursive-completion}
Let $N\ge1$, and consider a PSD decomposition
\begin{equation}
   X=\sum_{k\in\N}X_{(k)},
   \;
   X_{(k)}\succeq0,
   \label{eq:pf-total-split}
\end{equation}
such that, for every $k\in\N$ and every quantum operation by party $k$ with CJ matrix $C_k$,
\begin{equation}
   (X_{(k)})_{|C_k}
   \in\CS_{N-1}^{+,P\to F}.
   \label{eq:pf-total-cond-valid}
\end{equation}
Then the following two statements are equivalent:
\begin{enumerate}[label=(\roman*),leftmargin=2.4em]
\item The total operator is valid: $X\in\mathsf{Proc}_N^{+,P\to F}.$
\item Every first-level component is valid: $X_{(k)}\in\mathsf{Proc}_N^{+,P\to F}
   \;(k\in\N).$
\end{enumerate}
Consequently, the formulation in Definition~\ref{def:pf-direct-ecs}, which requires each first-level component to be valid, is equivalent to the following criterion requiring validity of the unconditioned sum:
\begin{equation}
\begin{aligned}
X\in\CS_N^{+,P\to F}
\;\text{iff}\;{}&
X\in\mathsf{Proc}_N^{+,P\to F},\\[-1mm]
&\exists\;
X=\sum_{k\in\N}X_{(k)},
\; X_{(k)}\succeq0,\\[-1mm]
&(X_{(k)})_{|C_k}\in\CS_{N-1}^{+,P\to F}
\;
\text{for all }k\in\N\text{ and all }C_k .
\end{aligned}
\label{eq:pf-CS-total-equivalent}
\end{equation}
\end{proposition}

\begin{proof}
(ii)$\Rightarrow$(i) follows immediately from
the convex-cone property of $\mathsf{Proc}_N^{+,P\to F}$.

For (i)$\Rightarrow$(ii), fix $j\in\N$ and a nonempty
$\mathcal X\subseteq\N\setminus\{j\}$.
Equation~\eqref{eq:pf-total-cond-valid} gives
\[
   (X_{(j)})_{|C_j}
   \in\CS_{N-1}^{+,P\to F}
   \subseteq
   \mathsf{Proc}_{N-1}^{+,P\to F}
   \;\forall C_j.
\]
Commuting the linear superoperator with conditioning gives
\[
\left[
\left(\prod_{i\in\mathcal X}
   (\mathcal I-\mathcal D_{A_O^i})\right)
\mathcal D_{A_{IO}^{(\N\setminus\{j\})\setminus\mathcal X}F}
X_{(j)}
\right]_{|C_j}=0
\;\forall C_j.
\]
The CJ matrices of quantum operations span the real vector space of Hermitian operators for party $j$, as shown in the proof of Proposition~\ref{prop:pf-first-component-order}. Since the identity above holds for every $C_j$, the operator in brackets must vanish, that is,
\begin{equation}
\left(\prod_{i\in\mathcal X}
   (\mathcal I-\mathcal D_{A_O^i})\right)
\mathcal D_{A_{IO}^{(\N\setminus\{j\})\setminus\mathcal X}F}
X_{(j)}=0
\;
(\varnothing\ne\mathcal X\subseteq\N\setminus\{j\}).
\label{eq:pf-component-derived-rest}
\end{equation}
Take any $k\in\N$. The validity constraint for the total operator
$X\in\mathsf{Proc}_N^{+,P\to F}$
with $\mathcal K=\{k\}$ gives
\begin{equation}
   0=
   (\mathcal I-\mathcal D_{A_O^k})
   \mathcal D_{A_{IO}^{\N\setminus\{k\}}F}X
   =
   \sum_{j\in\N}
   (\mathcal I-\mathcal D_{A_O^k})
   \mathcal D_{A_{IO}^{\N\setminus\{k\}}F}X_{(j)}.
   \label{eq:pf-auto-total-singleton}
\end{equation}
For every $j\ne k$, take
Eq.~\eqref{eq:pf-component-derived-rest} with
$\mathcal X=\{k\}\subseteq\N\setminus\{j\}$ and then apply
$\mathcal D_{A_{IO}^j}$ to obtain
$(\mathcal I-\mathcal D_{A_O^k}) \mathcal D_{A_{IO}^{\N\setminus\{k\}}F}X_{(j)}=0.$
Thus only the $j=k$ term in Eq.~\eqref{eq:pf-auto-total-singleton} can be nonzero, and hence
\begin{equation}
   (\mathcal I-\mathcal D_{A_O^k})
   \mathcal D_{A_{IO}^{\N\setminus\{k\}}F}X_{(k)}=0.
   \label{eq:pf-component-derived-first}
\end{equation}

To verify all validity constraints for $X_{(k)}$, take any $k\in\N$ and any nonempty $\mathcal K\subseteq\N$.
If $\mathcal K=\{k\}$, the required condition is
Eq.~\eqref{eq:pf-component-derived-first}.
If $k\notin\mathcal K$, use
Eq.~\eqref{eq:pf-component-derived-rest} for $j=k$ with
$\mathcal X=\mathcal K$, then apply
$\mathcal D_{A_{IO}^k}$.
If $k\in\mathcal K$ and $|\mathcal K|\ge2$, take
$\mathcal X=\mathcal K\setminus\{k\}$ in the same equation and apply
$\mathcal I-\mathcal D_{A_O^k}$.
Each of the three cases gives
\[
\left(\prod_{i\in\mathcal K}
   (\mathcal I-\mathcal D_{A_O^i})\right)
\mathcal D_{A_{IO}^{\N\setminus\mathcal K}F}X_{(k)}=0.
\]
Since $X_{(k)}\succeq0$, we obtain
$X_{(k)}\in\mathsf{Proc}_N^{+,P\to F}$.
Thus (ii) holds.

\end{proof}

\begin{corollary}[The single-party case and one-dimensional $P,F$]
\label{cor:pf-flat-equivalence}
The following two special cases hold:
\begin{enumerate}[label=(\roman*),leftmargin=2.4em]
\item For $N=1$, $ \CS_1^{+,P\to F}
   =
   \mathsf{Proc}_1^{+,P\to F}.$
\item For $P\simeq F\simeq\mathbb C$ and any $N$,
\begin{equation}
   \mathsf{Proc}_N^{+,P\to F}
   =
   \{Y\succeq0:Y\in\mathcal L^{\N}\},
   \;
   \CS_N^{+,P\to F}=\CSplus_N,
   \;
   \ECS_N^{P\to F}=\ECS_N.
   \label{eq:pf-auto-flat-equivalence}
\end{equation}
\end{enumerate}
\end{corollary}

\begin{proof}
For $N=1$, the sole component is $X$ and every conditioning leaves a PSD operator on $P\otimes F$, proving (i). For one-dimensional $P,F$, Eq.~\eqref{eq:basic-process-validity} reduces to the ordinary validity constraints. Proposition~\ref{prop:pf-first-component-order} supplies first-party compatibility for the general recursion; conversely, the ordinary recursion requires valid first-level components. Since both recursions have the same base case and conditioned-process condition, induction identifies their CS cones, including on extended inputs. The extension condition and scalar normalization then give $\ECS_N^{P\to F}=\ECS_N$.
\end{proof}

\label{sec:pf-proof-target}

By Proposition~\ref{prop:basic-pf-qccc}, necessity reduces to showing that
every $X\in\ECS_N^{+,P\to F}$ admits a PSD family
$\{T_\pi\}_{\pi\in\Pi_{\N}}$ such that
$X=\sum_{\pi\in\Pi_{\N}}T_\pi$ and all prefix sums satisfy the
ordered-sequence constraints. The isometric construction below combines the
recipient-dependent decompositions of Wechs et al.~\cite{wechs2019DefinitionCharacterisationMultipartite} into
the single compatible family required by QC-CC.

\begin{lemma}[The causally separable cone is closed and convex for fixed $P,F$]
For every $N\ge0$, $\CS_N^{+,P\to F}$ is a closed convex cone in the corresponding finite-dimensional space of Hermitian operators.
\label{lem:pf-CSclosed}
\end{lemma}

\begin{proof}
We proceed by induction on $N$. For $N=0$, $\CS_0^{+,P\to F}=\Pos(PF)$. Assume the claim holds for $N-1$.
Since $\mathsf{Proc}_N^{+,P\to F}$ is a closed convex cone defined by positivity and the finite family of homogeneous linear equalities
\eqref{eq:basic-process-validity}, if
$X=\sum_kX_{(k)}$ and $Y=\sum_kY_{(k)}$ each satisfy Definition~\ref{def:pf-direct-ecs}, then, for any $a,b\ge0$,
every component of $aX+bY=\sum_k(aX_{(k)}+bY_{(k)})$ still belongs to
$\mathsf{Proc}_N^{+,P\to F}$ and, after conditioning, belongs to
$\CS_{N-1}^{+,P\to F}$ by the induction hypothesis. This proves convexity.

Now let $X^{(m)}\to X$ with $X^{(m)}\in\CS_N^{+,P\to F}$. For each $m$, choose a witnessing decomposition
$X^{(m)}=\sum_{k\in\N}X_{(k)}^{(m)}, \; X_{(k)}^{(m)}\in\mathsf{Proc}_N^{+,P\to F}.$
Since $X^{(m)}\to X$ and $0\preceq X_{(k)}^{(m)}\preceq X^{(m)}$, the component sequences are bounded. In finite dimensions, a common subsequence can therefore be chosen such that
$X_{(k)}^{(m)}\longrightarrow X_{(k)}\succeq0, \; \sum_kX_{(k)}=X.$
$\mathsf{Proc}_N^{+,P\to F}$ is closed, so each limit $X_{(k)}$ remains in this cone. For a fixed CJ matrix $C_k$ of a quantum operation, conditioning is a continuous linear map, and
$\CS_{N-1}^{+,P\to F}$ is closed by the induction hypothesis. Hence
$(X_{(k)})_{|C_k} =\lim_m(X_{(k)}^{(m)})_{|C_k} \in\CS_{N-1}^{+,P\to F}.$
Taking the intersection over the CJ matrices of all quantum operations still gives a closed condition, so $X\in\CS_N^{+,P\to F}$.
\end{proof}

\begin{lemma}[Discarding input ancillas preserves the causally separable cone for fixed $P,F$]
\label{lem:pf-inputdiscard}
Suppose there are $N$ external slots and the input of party $j$ is extended to
$A_I^jE_jG_j$, where $G_j$ serves only as an input ancilla. If
$Y\in\CS_N^{+,P\to F}$ on these extended inputs, then
\begin{equation}
   \Tr_{G_1\cdots G_N}Y\in\CS_N^{+,P\to F}
   \label{eq:pf-inputdiscard}
\end{equation}
where the right-hand side retains only the inputs $A_I^jE_j$. The same holds when any subset of the $G_j$ is traced out.
\end{lemma}

\begin{proof}
The case $N=0$ holds by definition. For the induction step, choose $Y=\sum_kY_{(k)}$ with $Y_{(k)}\in\mathsf{Proc}_N^{+,P\to F}$ on the extended inputs. Partial trace preserves positivity. The identity $\Tr_{G_j}\mathcal D_{A_I^jE_jG_jA_O^j}=\mathcal D_{A_I^jE_jA_O^j}\Tr_{G_j}$, together with commutation on disjoint systems, gives, for every nonempty $\mathcal K\subseteq\N$,
\begin{align}
&\left(\prod_{i\in\mathcal K}(\mathcal I-\mathcal D_{A_O^i})\right)
 \mathcal D_{(\otimes_{j\in\N\setminus\mathcal K}A_I^jE_jA_O^j)F}
 \Tr_{G_1\cdots G_N}Y_{(k)}
\notag\\
&\;=
 \Tr_{G_1\cdots G_N}\!\left[
 \left(\prod_{i\in\mathcal K}(\mathcal I-\mathcal D_{A_O^i})\right)
 \mathcal D_{(\otimes_{j\in\N\setminus\mathcal K}A_I^jE_jG_jA_O^j)F}
 Y_{(k)}\right]=0.
 \label{eq:pf-discard-proc}
\end{align}
Each first-level component therefore remains in $\mathsf{Proc}_N^{+,P\to F}$ after the partial trace.

For a quantum operation $C_k$ on the reduced input, define
\begin{equation}
C_k^{\uparrow}:=C_k\otimes\id^{G_k}.
   \label{eq:pf-operation-lift}
\end{equation}
Since
$\Tr_{A_O^k}C_k^{\uparrow} =(\Tr_{A_O^k}C_k)\otimes\id^{G_k} \preceq\id^{A_I^kE_kG_k},$
this is still a valid CJ matrix of a quantum operation, and
\begin{equation}
   \left(\Tr_{G_1\cdots G_N}Y_{(k)}\right)_{|C_k}
   =
   \Tr_{\otimes_{j\ne k}G_j}
   \left((Y_{(k)})_{|C_k^{\uparrow}}\right).
   \label{eq:pf-discard-condition}
\end{equation}
The witness condition and the induction hypothesis put the right-hand side in $\CS_{N-1}^{+,P\to F}$, proving Eq.~\eqref{eq:pf-inputdiscard}.
\end{proof}

\begin{lemma}[Isometries preserve PSD decomposition components]
\label{lem:pf-support}\label{lem:dG-support}
If  $ Z=\sum_\alpha Z_\alpha,
   \; Z_\alpha\succeq0,$
then $\supp Z_\alpha\subseteq\supp Z$. If, in addition,
\begin{equation}
   Z=\beta\mathcal V X\mathcal V^\dagger,
   \; \beta>0,
   \; \mathcal V^\dagger\mathcal V=\id,
   \label{eq:pf-support-isom}
\end{equation}
then each $Z_\alpha$ has a unique representation
\begin{equation}
   Z_\alpha=\beta\mathcal V X_\alpha\mathcal V^\dagger,
   \;
   X_\alpha=\beta^{-1}\mathcal V^\dagger Z_\alpha\mathcal V\succeq0,
   \label{eq:pf-supportpull}
\end{equation}
and $X=\sum_\alpha X_\alpha$.
\end{lemma}

\begin{proof}
For $\ket\xi\in\Ker Z$, $0=\sum_\alpha\|Z_\alpha^{1/2}\ket\xi\|^2$ implies $Z_\alpha\ket\xi=0$, hence $\supp Z_\alpha\subseteq\supp Z$. Under Eq.~\eqref{eq:pf-support-isom}, every component is supported in $\Ran\mathcal V$. Inserting $\mathcal V\mathcal V^\dagger$ on both sides gives Eq.~\eqref{eq:pf-supportpull}; multiplication by $\mathcal V^\dagger,\mathcal V$ proves uniqueness, and summing gives $X=\sum_\alpha X_\alpha$.
\end{proof}

The following proposition uses the compactness and finite-intersection argument of Wechs et al.~\cite{wechs2019DefinitionCharacterisationMultipartite} to obtain a decomposition independent of the extension. The condition
$X_{(k)}\in\mathsf{Proc}_N^{+,P\to F}$ is defined by a finite family of closed linear constraints and positivity, so the same compactness argument applies for fixed $P,F$.
\begin{proposition}[Extension-independent first-level decomposition]
\label{lem:pf-uniform}
Let $N\ge1$ and $X\in\mathsf{Proc}_N^{+,P\to F}$. Then $X\in\ECS_N^{+,P\to F}$ if and only if there exist fixed operators
$\{X_{(k)}\}_{k\in\N}$ on the original process space such that
\begin{align}
   X&=\sum_{k\in\N}X_{(k)},
   \label{eq:pf-uniform-sum}\\
   X_{(k)}&\in\mathsf{Proc}_N^{+,P\to F}
   \;(k\in\N),
   \label{eq:pf-uniform-first}
\end{align}
These operators also satisfy the following condition: for arbitrary input extensions $E_1,\ldots,E_N$, any normalized shared state $\rho$, and any quantum operation of party $k$ with CJ matrix $C_k$,
\begin{equation}
   (X_{(k)}\otimes\rho)_{|C_k}
   \in\CS_{N-1}^{+,P\to F},
   \label{eq:pf-uniform-cond}
\end{equation}
where the right-hand side acts on the remaining labels $\N\setminus\{k\}$.

\end{proposition}

\begin{proof}
We first prove necessity. Suppose $X\in\ECS_N^{+,P\to F}$.

\smallskip\noindent\emph{Pure shared states.}
Fix finite-dimensional input extensions $E_1,\ldots,E_N$, write
$E:=E_1\otimes\cdots\otimes E_N$, and first suppose that the shared state is pure,
$\rho^{E_1\cdots E_N}=\proj\psi$. Since
$X\in\ECS_N^{+,P\to F}$, we have
$X\otimes\proj\psi\in\CS_N^{+,P\to F}$. We can therefore choose
\begin{equation}
   X\otimes\proj\psi
   =\sum_{k\in\N}Y_{(k)},
   \;
   Y_{(k)}\in\mathsf{Proc}_N^{+,P\to F},
   \label{eq:pf-pure-split}
\end{equation}
such that, for any quantum operation of party $k$ on the extended input with CJ matrix $C_k$,
$(Y_{(k)})_{|C_k}\in\CS_{N-1}^{+,P\to F}$.

We have $0\preceq Y_{(k)}\preceq X\otimes\proj\psi$. Define the isometric embedding
$\mathcal V_\psi:\mathcal H_0\to\mathcal H_0\otimes E$ by
$\mathcal V_\psi\ket\phi=\ket\phi\otimes\ket\psi$, where
$\mathcal H_0:=PA_{IO}^{\N}F$. Lemma~\ref{lem:pf-support} gives unique operators
$X_{(k)}^{[\psi]}\succeq0$ on the original process space such that
\begin{equation}
   Y_{(k)}
   =\mathcal V_\psi X_{(k)}^{[\psi]}\mathcal V_\psi^\dagger
   =X_{(k)}^{[\psi]}\otimes\proj\psi,
   \;
   X=\sum_{k\in\N}X_{(k)}^{[\psi]}.
   \label{eq:pf-pure-factor}
\end{equation}
Trace out all $E_j$ in each nonempty-subset constraint of Eq.~\eqref{eq:basic-process-validity}. If the completely depolarized block contains party $j$, then
$\Tr_{E_j}\circ\mathcal D_{A_I^jE_jA_O^j}
=\mathcal D_{A_I^jA_O^j}\circ\Tr_{E_j}$. Otherwise, the partial trace commutes with the output projections and depolarization of the other slots. Thus $X_{(k)}^{[\psi]}\in\mathsf{Proc}_N^{+,P\to F}.$
Moreover, $Y_{(k)}=X_{(k)}^{[\psi]}\otimes\proj\psi$ and the recursive condition give
\begin{equation}
   (X_{(k)}^{[\psi]}\otimes\proj\psi)_{|C_k}
   \in\CS_{N-1}^{+,P\to F}
   \;\forall k,\ \forall C_k.
   \label{eq:pf-pure-rec}
\end{equation}

\smallskip\noindent\emph{Mixed shared states.}
Fix a normalized state
$\rho^{E_1\cdots E_N}$, choose a finite-dimensional purification
$\proj\Psi^{E_1\cdots E_NG}$, and attach $G$ to the input of party $1$. Applying the preceding argument to the pure state
$\proj\Psi$ gives components $\{Z_k\}$ on the original space satisfying
$\sum_kZ_k=X$ and $Z_k\in\mathsf{Proc}_N^{+,P\to F}$, as well as the recursive condition for the CJ matrices of all local quantum operations on the purified extension. If $k=1$, take
$C_1^\uparrow=C_1\otimes\id^G$. If $k\ne1$, then, after conditioning, $G$ remains an input ancilla of the remaining party $1$. Lemma~\ref{lem:pf-inputdiscard} gives, in both cases,
\begin{equation}
   (Z_k\otimes\rho^{E_1\cdots E_N})_{|C_k}
   \in\CS_{N-1}^{+,P\to F}.
   \label{eq:pf-mixed-cond}
\end{equation}
Thus, for each fixed extension state, there is at least one first-level decomposition on the original space satisfying all the required conditions.

\smallskip\noindent\emph{A common decomposition for all extension states.}
We now remove the dependence of the first-level decomposition on the chosen extension state. Consider the finite-dimensional real vector space
\begin{equation}
   \mathsf X_N:=
   \{(Z_1,\ldots,Z_N):Z_k\in\Herm(PA_{IO}^{\N}F)\},
   \label{eq:pf-Kspace}
\end{equation}
and, for arbitrary finite-dimensional $E_1,\ldots,E_N$ and any normalized shared state
$\rho^{E_1\cdots E_N}$, let
$\mathfrak F_{\rho^{E_1\cdots E_N}}\subseteq\mathsf X_N$ be the set of all tuples satisfying
\begin{align}
   &Z_k\succeq0,
   \; \sum_{k\in\N}Z_k=X,
   \label{eq:pf-K1}\\
   &Z_k\in\mathsf{Proc}_N^{+,P\to F},
   \label{eq:pf-K2}\\
   &(Z_k\otimes\rho^{E_1\cdots E_N})_{|C_k}
   \in\CS_{N-1}^{+,P\to F}
   \;\forall k,\ \forall C_k
   \label{eq:pf-K3}
\end{align}
The preceding purification argument establishes nonemptiness. Since
$0\preceq Z_k\preceq X$, all these sets lie in the fixed compact set $\prod_{k\in\N}\{Z_k:0\preceq Z_k\preceq X\}$ in the finite-dimensional space $\mathsf X_N$. Positivity, the sum constraint,
the finite family of closed linear constraints defining $\mathsf{Proc}_N^{+,P\to F}$, and, for fixed
$\rho,k,C_k$, the inverse image of $\CS_{N-1}^{+,P\to F}$ under the continuous linear conditioning map are all closed conditions. Thus each
$\mathfrak F_\rho$ is compact.

We establish the finite intersection property as in Wechs et al.~\cite{wechs2019DefinitionCharacterisationMultipartite}. Take any finite collection of extension states
$\rho_t^{E_{1,t}\cdots E_{N,t}}$, assign to each party the ancillary input
$\bigotimes_tE_{j,t}$, and use the shared state $\bigotimes_t\rho_t$. Nonemptiness for this joint extension gives a family
$\{\widehat Z_k\}$ still acting on the original process space. Fix $t$ and lift the CJ matrix of the quantum operation of party $k$ to
$C_{k,t}\otimes\id^{\otimes_{u\ne t}E_{k,u}}$. Tracing out all remaining input ancilla factors with $u\ne t$ and applying Lemma~\ref{lem:pf-inputdiscard} gives
\begin{equation}
   (\widehat Z_k\otimes\rho_t)_{|C_{k,t}}
   \in\CS_{N-1}^{+,P\to F}.
   \label{eq:pf-finite-reduction}
\end{equation}
For each $t$, the same original-space components $\widehat Z_k\in\mathsf{Proc}_N^{+,P\to F}$ satisfy Eq.~\eqref{eq:pf-finite-reduction}. Hence the family
$\{\widehat Z_k\}$ belongs to all the finitely many sets $\mathfrak F_{\rho_t}$. This family of compact sets therefore has the finite intersection property, and hence
$\bigcap_{E_1,\ldots,E_N,\rho}\mathfrak F_\rho\ne\varnothing.$
Any tuple in the total intersection satisfies Eqs.~\eqref{eq:pf-uniform-sum}--\eqref{eq:pf-uniform-cond}.

\smallskip\noindent\emph{Sufficiency.}

Conversely, suppose fixed components $X_{(k)}$ satisfying these conditions exist. For any input extension and shared state $\rho$, we have $X\otimes\rho=\sum_kX_{(k)}\otimes\rho$. Tensoring with an input state preserves the process validity of each first-level component, and Eq.~\eqref{eq:pf-uniform-cond} then implies $X\otimes\rho\in\CS_N^{+,P\to F}$. Since the extension and $\rho$ are arbitrary, $X\in\ECS_N^{+,P\to F}$.
\end{proof}

\label{sec:pf-route}

For any fixed
$k\in\N$, we jointly encode the entire slot of party $k$ into the input ancillas of the other parties. Write
$d_k:=\dim A_{IO}^k$.
All tensor products over $\N\setminus\{k\}$ follow a fixed order, with any required permutations of tensor factors implicit in the notation. Fix
an orthonormal basis of $A_{IO}^k$, denoted by $\{\ket a_{A_{IO}^k}\}_{a=1}^{d_k}$. For each $r\in\N\setminus\{k\}$, add to the input of party $r$ the system
\begin{equation}
   G_r:=\mathbb C\ket{\mathrm{vac}}_{G_r}\oplus S_r,
   \;
   S_r:=\operatorname{span}\{\ket a_{S_r}\}_{a=1}^{d_k},
   \label{eq:pf-LG}
\end{equation}
where all the indicated bases are orthonormal and $S_r\perp\mathbb C\ket{\mathrm{vac}}_{G_r}$. In this construction, we call $S_r$ the data sector and $\mathbb C\ket{\mathrm{vac}}_{G_r}$ the orthogonal vacuum sector. For each $u\in\N\setminus\{k\}$, let
\begin{equation}
   J_u:A_{IO}^k\longrightarrow G_u,
   \; J_u\ket a_{A_{IO}^k}:=\ket a_{S_u}\;(a=1,\ldots,d_k),
   \label{eq:pf-Ju-explicit}
\end{equation}
and define
\begin{equation}
   V_u:A_{IO}^k\longrightarrow\bigotimes_{r\in\N\setminus\{k\}}G_r,
   \;
   V_u\ket\phi_{A_{IO}^k}
   :=J_u\ket\phi_{A_{IO}^k}
      \otimes\bigotimes_{r\in\N\setminus\{k,u\}}\ket{\mathrm{vac}}_{G_r}.
   \label{eq:pf-Vu}
\end{equation}
For $u\ne v$, on $G_u$, the range of $V_u$ lies in the data sector $S_u$, whereas that of $V_v$ lies in the vacuum sector. Together with $J_u^\dagger J_u=\id^{A_{IO}^k}$, this gives
\begin{equation}
   V_u^\dagger V_v=\delta_{uv}\id^{A_{IO}^k},
   \;
   V:=({N-1})^{-1/2}\sum_{u\in\N\setminus\{k\}}V_u,
   \; V^\dagger V=\id^{A_{IO}^k}.
   \label{eq:pf-V}
\end{equation}
$V$ coherently superposes the encodings $V_u$, each placing $A_{IO}^k$ in party $u$'s data sector and all other $G_r$ in the vacuum state.
Take two mutually independent ancillary copies isomorphic to the actual process slot $A_{IO}^k$,
$\widetilde A_{IO}^k\simeq A_{IO}^k$ and $S\simeq A_{IO}^k$. Fix orthonormal bases of these spaces corresponding to the basis of the actual slot. Transport $V$ to the ancillary copy using this fixed basis isomorphism, and denote the resulting isometry by
$\widetilde V:\widetilde A_{IO}^k\to\bigotimes_{r\in\N\setminus\{k\}}G_r$. $V$ continues to denote the isometry on the actual slot defined in Eq.~\eqref{eq:pf-V}.
Write
\begin{equation}
   \ket{\Phi_{d_k}}_{S\,\widetilde A_{IO}^k}
   :=d_k^{-1/2}\sum_{a=1}^{d_k}\ket a_S\ket a_{\widetilde A_{IO}^k},
   \; 
   \ket{\Omega_V}_{S\,\bigotimes_{r\in\N\setminus\{k\}}G_r}
   :=(\id^S\otimes\widetilde V)\ket{\Phi_{d_k}}_{S\,\widetilde A_{IO}^k}.
\label{eq:pf-Omega}
\end{equation}
In the same corresponding bases, we also use $\ket{\Phi_{d_k}}_{S A_{IO}^k}$ to denote the maximally entangled state of $S$ and the actual slot $A_{IO}^k$.
Below, $S$ serves as an additional input ancilla of party $k$, and $G_r$ as an input ancilla of party $r\in\N\setminus\{k\}$.
We write
\begin{equation}
   \mathcal V:=\id^P\otimes V\otimes\id^{A_{IO}^{\N\setminus\{k\}}F}.
   \label{eq:pf-bigV}
\end{equation}
\begin{lemma}[Conditioning on a maximally entangled projection yields an ECS process for the remaining $N-1$ parties]
\label{lem:pf-route}\label{lem:dG-teleport}
On the finite-dimensional systems specified above, let $X\in\ECS_N^{+,P\to F}$ and let
$\{X_{(j)}\}_{j\in\N}$ be the
first-level decomposition supplied by Proposition~\ref{lem:pf-uniform}, namely,
\begin{align}
   X&=\sum_{j\in\N}X_{(j)},
   \label{eq:pf-route-assumption-sum}\\
   X_{(j)}&\in\mathsf{Proc}_N^{+,P\to F}
   \;(j\in\N).
   \label{eq:pf-route-assumption-first}
\end{align}
and, for arbitrary finite-dimensional input ancillas $E_1,\ldots,E_N$, any normalized shared state
$\rho^{E_1\cdots E_N}$, any $j\in\N$, and any quantum operation of party $j$ with CJ matrix $C_j$, we have
\begin{equation}
   (X_{(j)}\otimes\rho)_{|C_j}\in\CS_{N-1}^{+,P\to F}.
   \label{eq:pf-route-assumption-uniform-cond}
\end{equation}
For any $k$ such that $X_{(k)}\ne0$, take, under the
fixed identification of tensor factors $A_I^kSA_O^k\simeq S A_{IO}^k$,
\begin{equation}
   C_k:=\lambda\proj{\Phi_{d_k}}_{S A_{IO}^k},\; 0<\lambda\le1.
   \label{eq:pf-me-operation}
\end{equation}
Then
\begin{equation}
   Z_k:=(X_{(k)}\otimes\proj{\Omega_V})_{|C_k}
   =\beta\,\mathcal V X_{(k)}\mathcal V^\dagger,
   \; \beta:=\lambda/d_k^2>0,
   \label{eq:pf-teleport}
\end{equation}
and, for the remaining slots with input spaces $A_I^rG_r$ and output spaces
$A_O^r$ ($r\in\N\setminus\{k\}$), forming an $N-1$-slot process, we have
\begin{equation}
   Z_k\in\ECS_{N-1}^{+,P\to F}.
   \label{eq:pf-Z-ECS}
\end{equation}
\end{lemma}

\begin{proof}
Since $C_k\succeq0$ and $\Tr_{A_O^k}\proj{\Phi_{d_k}}$ is PSD with trace $1$, $0<\lambda\le1$ implies $\Tr_{A_O^k}C_k\preceq\id^{A_I^kS}$. Thus $C_k$ is an allowed quantum operation.

We now evaluate the conditioning on the maximally entangled projection on $S A_{IO}^k$. Let $H$ be any system not involved in this projection, and write $Y=\sum_{a,b=1}^{d_k}\oprod a b_{A_{IO}^k}\otimes Y_{ab}.$
We use the expansions
\begin{align}
   \proj{\Phi_{d_k}}_{S A_{IO}^k}
   &=\frac1{d_k}\sum_{\mu,\nu=1}^{d_k}
     \oprod\mu\nu_S\otimes\oprod\mu\nu_{A_{IO}^k},
   \label{eq:pf-Phi-expand}\\
   \proj{\Omega_V}_{S\,\bigotimes_{r\in\N\setminus\{k\}}G_r}
   &=\frac1{d_k}\sum_{i,j=1}^{d_k}
     \oprod i j_S\otimes \widetilde V\oprod i j_{\widetilde A_{IO}^k}\widetilde V^\dagger,
   \label{eq:pf-Omega-expand}
\end{align}
Under the fixed basis isomorphism, $\widetilde V\oprod i j_{\widetilde A_{IO}^k}\widetilde V^\dagger
=V\oprod i j_{A_{IO}^k}V^\dagger$. Applying
$\Tr_{A_{IO}^k}(\oprod\mu\nu\oprod a b)=\delta_{\nu a}\delta_{\mu b}$ and
$\Tr_S(\oprod\mu\nu\oprod i j)=\delta_{\nu i}\delta_{\mu j}$ yields
\begin{equation}
   \Tr_{S A_{IO}^k}\!\left[(C_k\otimes\id)
   (Y\otimes\proj{\Omega_V})\right]
   =\frac{\lambda}{d_k^2}(V\otimes\id^H)Y(V^\dagger\otimes\id^H).
   \label{eq:pf-me-transport}
\end{equation}
Taking $H=P\otimes A_{IO}^{\N\setminus\{k\}}\otimes F$ and $Y=X_{(k)}$ gives Eq.~\eqref{eq:pf-teleport}.

For arbitrary further input ancillas $H_r$ and a normalized shared state $\tau$, apply Eq.~\eqref{eq:pf-route-assumption-uniform-cond} with
\begin{equation}
   E_k:=S,\;
   E_r:=G_r\otimes H_r\ (r\in\N\setminus\{k\}),
   \;
   \rho:=\proj{\Omega_V}\otimes\tau,
   \label{eq:pf-Omega-tau}
\end{equation}
and the same operation $C_k$. Equation~\eqref{eq:pf-me-transport} gives $Z_k\otimes\tau\in\CS_{N-1}^{+,P\to F}$. Since the $H_r$ and $\tau$ are arbitrary, Eq.~\eqref{eq:pf-Z-ECS} follows.
\end{proof}

\begin{lemma}[Compression by the adjoint isometry preserves the QC-CC linear constraints]
\label{lem:pf-compress}
Use the construction of Lemma~\ref{lem:pf-route}. Fix the first party $k$ and choose
a nonempty ordered sequence in $\N\setminus\{k\}$, denoted by $h=(r_1,\ldots,r_n).$ 
Let
\begin{equation}
   \mathcal V_u:=\id^P\otimes V_u\otimes\id^{A_{IO}^{\N\setminus\{k\}}F}\;(u\in\N\setminus\{k\}),
   \; \mathcal V_{r_1}:=\mathcal V_u\big|_{u=r_1}.
   \label{eq:pf-Vr1-big}
\end{equation}
Then, for every
$Y\in\Herm\!\left(P\otimes A_{IO}^{\mathcal N}\otimes F\right)$,
\begin{equation}
\begin{aligned}
&\mathcal V_{r_1}^\dagger
(\mathcal I-\mathcal D_{A_O^{r_n}})
\mathcal D_{(\otimes_{j\in\N\setminus\{k,r_1,\ldots,r_n\}}A_{IO}^jG_j)F}
\bigl(\mathcal VY\mathcal V^\dagger\bigr)
\mathcal V_{r_1}
\\[1mm]
&\;=
\frac{\displaystyle\prod_{j\in\N\setminus\{k,r_1,\ldots,r_n\}}d_{G_j}^{-1}}{{N-1}}
(\mathcal I-\mathcal D_{A_O^{r_n}})
\mathcal D_{A_{IO}^{\N\setminus\{k,r_1,\ldots,r_n\}}F}Y.
\end{aligned}
   \label{eq:pf-pullback}
\end{equation}
When $\N\setminus\{k,r_1,\ldots,r_n\}=\varnothing$, the product above equals $1$ by convention.
\end{lemma}

\begin{proof}
Using $V=(N-1)^{-1/2}\sum_{a\in\N\setminus\{k\}}V_a$, we expand
\begin{equation}
   \mathcal VY\mathcal V^\dagger
   =\frac{1}{N-1}\sum_{a,b\in\N\setminus\{k\}}\mathcal V_aY\mathcal V_b^\dagger.
   \label{eq:pf-route-cross}
\end{equation}

First consider the off-diagonal blocks $a\ne b$. If $a\in\N\setminus\{k,r_1,\ldots,r_n\}$, then, on $G_a$,
$\mathcal V_aY\mathcal V_b^\dagger$ has left support in the data sector $S_a$ and right support in the vacuum sector, so its partial trace over $G_a$ vanishes. If $b\in\N\setminus\{k,r_1,\ldots,r_n\}$, the partial trace over $G_b$ likewise vanishes. Thus
\begin{equation}
\mathcal D_{(\otimes_{j\in\N\setminus\{k,r_1,\ldots,r_n\}}A_{IO}^jG_j)F}
   (\mathcal V_aY\mathcal V_b^\dagger)=0
   \;
   (a\ne b,\ \{a,b\}\cap(\N\setminus\{k,r_1,\ldots,r_n\})\ne\varnothing).
   \label{eq:pf-route-offdiag-future}
\end{equation}
If $a,b\notin\N\setminus\{k,r_1,\ldots,r_n\}$, every completely depolarized $G_j$ lies in the vacuum sector in the isometric components $a,b,r_1$. These systems contribute only positive scalar factors, while the isometric ranges on the systems that are not depolarized satisfy
$V_{r_1}^\dagger V_a=\delta_{r_1a}\id^{A_{IO}^k}$ and
$V_b^\dagger V_{r_1}=\delta_{br_1}\id^{A_{IO}^k}$. Since $a\ne b$, these two factors cannot both be nonzero. Hence
\begin{equation}
   \mathcal V_{r_1}^\dagger
   \mathcal D_{(\otimes_{j\in\N\setminus\{k,r_1,\ldots,r_n\}}A_{IO}^jG_j)F}
   (\mathcal V_aY\mathcal V_b^\dagger)
   \mathcal V_{r_1}=0
   \;(a\ne b,\ a,b\notin\N\setminus\{k,r_1,\ldots,r_n\}).
   \label{eq:pf-route-offdiag-called}
\end{equation}
Equations~\eqref{eq:pf-route-offdiag-future}--\eqref{eq:pf-route-offdiag-called} therefore eliminate all off-diagonal encoding blocks.

Next consider the diagonal blocks $a=b\ne r_1$. If $a\notin\N\setminus\{k,r_1,\ldots,r_n\}$, then $G_a$ is not depolarized by the future depolarization map. On this system, isometric component $a$ lies in the data sector and component $r_1$ in the vacuum sector, so compression on both sides annihilates the block. If $a\in\N\setminus\{k,r_1,\ldots,r_n\}$, then $G_{r_1}$ is not depolarized by that map. On this system, component $a$ lies in the vacuum sector and component $r_1$ in the data sector, so compression on both sides again annihilates the block. The only surviving block in Eq.~\eqref{eq:pf-route-cross} is therefore $(a,b)=(r_1,r_1)$.

In this block, for every $j\in\N\setminus\{k,r_1,\ldots,r_n\}$, $G_j$ is in the vacuum state. Applying to $G_j$ the map $\mathcal D_{G_j}$ and then taking the matrix element with the vacuum factor of $V_{r_1}$ gives $\bra{\mathrm{vac}}\frac{\id^{G_j}}{d_{G_j}}\ket{\mathrm{vac}}
   =\frac1{d_{G_j}}.$
Depolarizing maps on distinct systems commute, and $\mathcal V_{r_1}$ acts as the identity on the original systems
$A_{IO}^{\N\setminus\{k,r_1,\ldots,r_n\}}F$. Therefore
\begin{equation}
\begin{aligned}
&\mathcal V_{r_1}^\dagger
  \mathcal D_{(\otimes_{j\in\N\setminus\{k,r_1,\ldots,r_n\}}A_{IO}^jG_j)F}
  (\mathcal VY\mathcal V^\dagger)
  \mathcal V_{r_1}
\\
&\;=
\frac1{N-1}\left(\prod_{j\in\N\setminus\{k,r_1,\ldots,r_n\}}\frac1{d_{G_j}}\right)
\mathcal D_{A_{IO}^{\N\setminus\{k,r_1,\ldots,r_n\}}F}Y.
\end{aligned}
   \label{eq:pf-joint-depol}
\end{equation}
When $\N\setminus\{k,r_1,\ldots,r_n\}=\varnothing$, the product equals $1$ by convention, and the same conclusion follows directly from
$V_{r_1}^\dagger V_a=\delta_{r_1a}\id^{A_{IO}^k}$, which selects the unique surviving block.

Finally, $A_O^{r_n}$ is disjoint from the systems depolarized by the future depolarization map, and $\mathcal V_{r_1}$ acts as the identity on
$A_O^{r_n}$. We can therefore commute
$\mathcal I-\mathcal D_{A_O^{r_n}}$ through the compression on both sides and apply it to
Eq.~\eqref{eq:pf-joint-depol}, obtaining Eq.~\eqref{eq:pf-pullback}.
\end{proof}

\begin{proposition}[Inductive construction of a common decomposition into components associated with complete orders]
\label{prop:pf-Hm}
For any $N\ge1$, 
if $X\in\ECS_N^{+,P\to F},$
then there exist PSD operators indexed by each full permutation $\pi\in\Pi_{\N}$, denoted by $T_\pi$, such that $X=\sum_{\pi\in\Pi_{\N}}T_\pi,$
and, for every nonempty ordered sequence $h$,
\begin{equation}
   (\mathcal I-\mathcal D_{A_O^{\ell(h)}})
   \mathcal D_{A_{IO}^{\N\setminus\{k_1,\ldots,k_n\}}F}
   \left(\sum_{\pi\succeq h}T_\pi\right)=0.
   \label{eq:pf-H-prefix}
\end{equation}
\end{proposition}

\begin{proof}
We proceed by induction on $N$.

Denote the sole external label by $k$. By Definition~\ref{def:pf-direct-ecs}, the only first-level component is $X$, and
$X\in\mathsf{Proc}_1^{+,P\to F}$. Taking in Eq.~\eqref{eq:basic-process-validity}
$\mathcal K=\{k\}$ gives
$(\mathcal I-\mathcal D_{A_O^k})\mathcal D_FX=0.$
Setting the PSD operator corresponding to the only complete order to $T_{(k)}:=X$ proves the claim.

Let $N\ge2$ and assume the proposition holds for arbitrary input-output spaces and $P,F$ with $N-1$ external operations. Thus every
$\ECS_{N-1}^{+,P\to F}$ process admits an expansion in PSD operators indexed by complete orders and satisfying the corresponding sum and prefix constraints. For
$X\in\ECS_N^{+,P\to F}$,
Proposition~\ref{lem:pf-uniform} gives a fixed decomposition
\begin{equation}
   X=\sum_{k\in\N}X_{(k)}.
   \label{eq:pf-ind-firstsplit}
\end{equation}
Fix $k$. If $X_{(k)}\ne0$, Lemma~\ref{lem:pf-route} gives
\begin{equation}
   Z_k=\beta\mathcal V X_{(k)}\mathcal V^\dagger
   \in\ECS_{N-1}^{+,P\to F}.
   \label{eq:pf-ind-Z}
\end{equation}
Here $\ECS_{N-1}^{+,P\to F}$ acts on the remaining slots labeled by $\N\setminus\{k\}$, of which there are $N-1$.
By the induction hypothesis, prepending $k$ to each permutation of the remaining parties gives components indexed by full permutations $\pi\in\Pi_\N$ with $\pi\succeq(k)$ such that
\begin{equation}
   Z_k=\sum_{\pi\succeq(k)}Z_\pi,
   \; Z_\pi\succeq0,
   \label{eq:pf-ind-Zcomponents}
\end{equation}
and, for every nonempty ordered sequence in $\N\setminus\{k\}$,
$h=(r_1,\ldots,r_n)$, we have
\begin{equation}
   (\mathcal I-\mathcal D_{A_O^{r_n}})
   \mathcal D_{(\otimes_{j\in\N\setminus\{k,r_1,\ldots,r_n\}}A_{IO}^jG_j)F}
   \left(\sum_{\pi\succeq(k,h)}Z_\pi\right)=0.
   \label{eq:pf-ind-expanded-prefix}
\end{equation}
Since $0\preceq Z_\pi\preceq Z_k$,
Lemma~\ref{lem:pf-support} gives, for each $\pi\succeq(k)$, a unique $X_\pi$ such that
\begin{equation}
   Z_\pi=\beta\mathcal V X_\pi\mathcal V^\dagger,
   \; X_\pi\succeq0,
   \;
   X_{(k)}=\sum_{\pi\succeq(k)}X_\pi.
   \label{eq:pf-ind-pull-components}
\end{equation}
Define the partial sum specified by the entire prefix $(k,h)$ as $X_{(k,h)}:=\sum_{\pi\succeq(k,h)}X_\pi$.
Then
\begin{equation}
   \sum_{\pi\succeq(k,h)}Z_\pi
   =\beta\mathcal V X_{(k,h)}\mathcal V^\dagger.
   \label{eq:pf-Zh-Xh}
\end{equation}
Substitute this into Eq.~\eqref{eq:pf-ind-expanded-prefix} and compress by applying
$\mathcal V_{r_1}^\dagger,\mathcal V_{r_1}$ on the left and right, respectively. Lemma~\ref{lem:pf-compress} gives
\begin{equation}
   0=
   \frac{\beta}{|\N\setminus\{k\}|} \displaystyle\prod_{j\in\N\setminus\{k,r_1,\ldots,r_n\}}d_{G_j}^{-1}
   (\mathcal I-\mathcal D_{A_O^{r_n}})
   \mathcal D_{A_{IO}^{\N\setminus\{k,r_1,\ldots,r_n\}}F}X_{(k,h)}.
   \label{eq:pf-compress-zero}
\end{equation}
Since the prefactor is strictly positive, we obtain
\begin{equation}
   (\mathcal I-\mathcal D_{A_O^{r_n}})
   \mathcal D_{A_{IO}^{\N\setminus\{k,r_1,\ldots,r_n\}}F}X_{(k,h)}=0.
   \label{eq:pf-original-local-prefix}
\end{equation}
Equation~\eqref{eq:pf-original-local-prefix} gives the constraint for the global ordered sequence $(k,r_1,\ldots,r_n)$. By Proposition~\ref{lem:pf-uniform},
$X_{(k)}\in\mathsf{Proc}_N^{+,P\to F}$. Taking in Eq.~\eqref{eq:basic-process-validity}
$\mathcal K=\{k\}$ gives the ordered-sequence constraint of length $1$,
\begin{equation}
   (\mathcal I-\mathcal D_{A_O^k})
   \mathcal D_{A_{IO}^{\N\setminus\{k\}}F}X_{(k)}=0.
   \label{eq:pf-length1-prefix}
\end{equation}

If $X_{(k)}=0$, simply set, for all $\pi\succeq(k)$, $X_\pi=0$. All corresponding ordered-sequence constraints then hold automatically.
Repeat this construction for all $k\in\N$ and define, for each full permutation, $T_\pi:=X_\pi\succeq0.$
Equations~\eqref{eq:pf-ind-firstsplit} and \eqref{eq:pf-ind-pull-components}, together with the preceding convention for zero components, give
\begin{equation}
   \sum_{\pi\in\Pi_{\N}}T_\pi=\sum_{k\in\N}\sum_{\pi\succeq(k)}X_\pi
   =\sum_{k\in\N}X_{(k)}=X.
   \label{eq:pf-global-permutation-sum}
\end{equation}
Equation~\eqref{eq:pf-length1-prefix} covers ordered sequences of length $1$,
and Eq.~\eqref{eq:pf-original-local-prefix} covers all longer ordered sequences. Hence
Eq.~\eqref{eq:pf-H-prefix} holds for every nonempty ordered sequence.
\end{proof}

\begin{proof}[Proof of Theorem~\ref{thm:pf-ECS-QCCC}]
(i)$\Longrightarrow$(ii) follows from Proposition~\ref{prop:pf-Hm} and the QC-CC characterization in Proposition~\ref{prop:basic-pf-qccc}. We prove (ii)$\Longrightarrow$(i).

Fix the components of $W\in\QCCC_N^{P\to F}$ corresponding to complete orders, denoted by $T_\pi$. Choose arbitrary input ancillas
$E_1,\ldots,E_N$ and a normalized shared state $\rho$. Incorporate $\rho$ into the initial internal map and transmit each
$E_k$ through the circuit, obtaining the extended process $W\otimes\rho=\sum_\pi(T_\pi\otimes\rho)$. This extension also follows from the reduced-history constraints: for $h=(k_1,\ldots,k_n)$ set
\[
\widetilde Q_h:=Q_h\otimes\rho_{E_{k_1}\cdots E_{k_n}},
\qquad \widetilde T_\pi:=T_\pi\otimes\rho,
\]
where the subscript on $\rho$ denotes its reduced state. Tracing out $A_I^jE_j$ at a successor history uses $\Tr_{E_j}\rho_{E_{k_1}\cdots E_{k_n}E_j}=\rho_{E_{k_1}\cdots E_{k_n}}$. Hence the initial, recursive, and terminal equations in Proposition~\ref{prop:basic-pf-qccc} are preserved.
Under the conditioning convention of Eq.~\eqref{eq:basic-condition}, each fixed $C_k$ corresponds to inserting the allowed quantum operation with CJ matrix $C_k^T$ into this circuit.

For a sequence of parties already called, $h=(k_1,\ldots,k_r)$, and the corresponding quantum operations
$\boldsymbol C_h=(C_{k_1},\ldots,C_{k_r})$, write
\begin{equation}
Y_{h,\boldsymbol C_h}:=
\sum_{\pi\succeq h}(T_\pi\otimes\rho)_{|\boldsymbol C_h},
\;
m:=|\N\setminus\{k_1,\ldots,k_r\}|=N-r.
\label{eq:pf-qccc-ecs-future}
\end{equation}
This operator is PSD. After arbitrary CPTP CJ matrices $M_j$ are inserted into the remaining slots, each QC-CC layer sums over all
future classical outcomes and preserves the trace. Thus there exists, depending only on $(h,\boldsymbol C_h)$, an operator
$Q_{h,\boldsymbol C_h}\in\Pos(P)$ such that
\begin{equation}
\Tr_F\!\left[
   \left(\bigotimes_{j\in\N\setminus\{k_1,\ldots,k_r\}}M_j\right)
   *Y_{h,\boldsymbol C_h}
\right]
=Q_{h,\boldsymbol C_h}
\;\text{for all CPTP }M_j.
\label{eq:pf-qccc-ecs-trace-effect-rewrite}
\end{equation}
Proposition~\ref{prop:pf-validity} therefore gives
$Y_{h,\boldsymbol C_h}\in\mathsf{Proc}_m^{+,P\to F}$.

We proceed by induction on $m$. For $m=0$, $Y_{h,\boldsymbol C_h}\in\Pos(P\otimes F)=
\CS_0^{+,P\to F}$. If $m\ge1$, let
\[
Y_{h,\boldsymbol C_h}^{(j)}:=
\sum_{\pi\succeq(h,j)}(T_\pi\otimes\rho)_{|\boldsymbol C_h},
\;
Y_{h,\boldsymbol C_h}=\sum_{j\in\N\setminus\{k_1,\ldots,k_r\}}
Y_{h,\boldsymbol C_h}^{(j)}.
\]
Every component is PSD, and
$(Y_{h,\boldsymbol C_h}^{(j)})_{|C_j}=Y_{(h,j),(\boldsymbol C_h,C_j)}$.
The induction hypothesis places the latter in
$\CS_{m-1}^{+,P\to F}$. Using the validity of the total operator, Proposition~\ref{prop:pf-recursive-completion} then gives
$Y_{h,\boldsymbol C_h}\in\CS_m^{+,P\to F}$. Taking $h=\varnothing$ yields
$W\otimes\rho\in\CS_N^{+,P\to F}$. Since the extension and $\rho$ are arbitrary,
$W\in\ECS_N^{P\to F}$.
\end{proof}

\section{Consistency with the ECS definition for one-dimensional $P,F$}
\label{sec:pf-consistency}
This section proves that, when the same operator is viewed as an ordinary
$(N+2)$-party process, the direct recursive definition for general
$P,F$ agrees exactly with the ECS definition for one-dimensional $P,F$.

Represent the global past $P$ and global future $F$ as two ordinary parties $\mathsf P,\mathsf F$, respectively:
\begin{equation}
   A_I^{\mathsf P}\simeq\mathbb C,\; A_O^{\mathsf P}=P,
   \;
   A_I^{\mathsf F}=F,\; A_O^{\mathsf F}\simeq\mathbb C.
   \label{eq:pf-participant-representation}
\end{equation}
Then
\begin{equation}
   A_I^{\mathsf P}\otimes A_O^{\mathsf P}\otimes
   \bigotimes_{k\in\N}(A_I^k\otimes A_O^k)
   \otimes A_I^{\mathsf F}\otimes A_O^{\mathsf F}
   \simeq
   P\otimes A_{IO}^{\N}\otimes F.
   \label{eq:pf-participant-space}
\end{equation}
The same CJ matrix admits two interpretations: $W$ is a $P\to F$ supermap, and $\widehat W$ is an ordinary $(N+2)$-party process matrix.
For this CJ correspondence between the global-past/global-future representation and ordinary multipartite processes, see Wechs et al.~\cite{wechs2021QuantumCircuitsClassical}.

With the slot set $\N\cup\{\mathsf P,\mathsf F\}$ and one-dimensional global systems outside the slots, Eq.~\eqref{eq:basic-process-validity} gives the ordinary process-validity conditions for $\widehat W$. The original $P,F$ remain the output of $\mathsf P$ and the input of $\mathsf F$, respectively. Party $\mathsf P$, with a one-dimensional input, is compatible with acting first, and party $\mathsf F$, with a one-dimensional output, is compatible with acting last~\cite{wechs2019DefinitionCharacterisationMultipartite}. If $\widehat W\in\QCCC_{N+2}$, Proposition~\ref{prop:pf-qccc-endpoint-reordering} below further guarantees a circuit realization in which every nonzero complete call sequence starts with $\mathsf P$ and ends with $\mathsf F$. Applying the ordinary multipartite ECS definition to $\widehat W$ amounts to testing
\begin{equation}
   \widehat W\in\ECS_{N+2},
   \label{eq:pf-flat-ecs-Nplus2}
\end{equation}
where $\ECS_{N+2}$ uses precisely the ordinary multipartite definition of Sec.~\ref{sec:foundations}~\cite{wechs2019DefinitionCharacterisationMultipartite}. We retain the notation $\ECS_{N+2}$ for this representation.

\begin{proposition}[Reordering the $P,F$ parties to the first and last positions]
\label{prop:pf-qccc-endpoint-reordering}
\label{lem:pf-qccc-endpoint-reordering}
Let $N\ge1$. Under the natural identification in Eqs.~\eqref{eq:pf-participant-representation}--\eqref{eq:pf-participant-space},
\begin{equation}
   \widehat W\in\QCCC_{N+2}
   \;\text{iff}\;
   W\in\QCCC_N^{P\to F}.
   \label{eq:pf-qccc-participant-equivalence}
\end{equation}
The left-hand side treats
$(A_I^{\mathsf P},A_O^{\mathsf P})=(\mathbb C,P)$ and
$(A_I^{\mathsf F},A_O^{\mathsf F})=(F,\mathbb C)$
as two parties. If the left-hand side holds, the same $\widehat W$ admits a QC-CC realization in which every nonzero complete call sequence has the form
\begin{equation}
   (\mathsf P,\pi_1,\ldots,\pi_N,\mathsf F),
   \; \pi\in\Pi_{\N}.
   \label{eq:pf-endpoint-order-sequence}
\end{equation}
\end{proposition}

\begin{proof}
First suppose that $\widehat W\in\QCCC_{N+2}$ and fix an $(N+2)$-party QC-CC realization. Write $\widehat\N:=\{\mathsf P,\mathsf F\}\cup\N$, and denote the input and output spaces of each $x\in\widehat\N$ by $A_I^x,A_O^x$, respectively. After completing $q$ external slots, denote the internal memory by $\alpha_q$. For an incomplete history $h$ of length $q\ge1$, write the Choi matrix of the corresponding internal instrument element as
\[
\widehat M_h^{\to y}\in\Pos\!\bigl(A_O^{\ell(h)}\alpha_q A_I^y\alpha_{q+1}\bigr).
\]
At the root, $\widehat M_{\varnothing}^{\to y}\in\Pos(A_I^y\alpha_1)$, and the final internal map after a complete history $h$ is $\widehat M_h^{\rm fin}\in\Pos(A_O^{\ell(h)}\alpha_{N+2})$. In the reordering construction below, we also use $\pi=(\pi_1,\ldots,\pi_q)$ for a partial history that does not contain $x$, and write
\[
\pi[x.r]:=(\pi_1,\ldots,\pi_r,x,\pi_{r+1},\ldots,\pi_q),
\; 0\le r\le q.
\]
The proof has three steps.

\smallskip
\noindent\emph{(a) If $A_I^s\simeq\mathbb C$, then $s$ can be moved to the first position in every nonzero complete history.}
For each $j=2,\ldots,N+2$, take a copy $\bar A_{O,j}^s\simeq A_O^s$, and for each $q=2,\ldots,N+2$ and $0\le r\le q-2$, take a copy $\alpha_q^{(r)}\simeq\alpha_q$. The stage memories of the new circuit are
\begin{equation}
\eta_1:=\mathbb C,
\;
\eta_{q+1}:=\alpha_q\bar A_{O,q+1}^s
\oplus\bigoplus_{r=0}^{q-1}\alpha_{q+1}^{(r)}
\;(1\le q\le N+1).
\label{eq:pf-firstslot-eta-new}
\end{equation}
The first term is the pre-$s$ sector, which stores the output of the actual first call to $s$, namely $A_O^s$. The $r$th direct-sum term is the post-$s$ sector with index $r$, corresponding to the original history $\pi[s.r]$. The new root instrument has only the outcome $s$. Set $\widetilde M_{\varnothing}^{\to s}=1$. After calling $s$, for each $y\ne s$, the operator $\widetilde M_{(s)}^{\to y}\in\Pos(A_O^sA_I^y\eta_2)$ has only two types of nonzero blocks:
\begin{align}
\left.\widetilde M_{(s)}^{\to y}\right|_{A_O^s A_I^y\alpha_1\bar A_{O,2}^s}
&=|\id\rangle\!\rangle\langle\!\langle\id|_{A_O^s\bar A_{O,2}^s}
  \otimes\widehat M_{\varnothing}^{\to y},
\label{eq:pf-firstslot-initial-direct}\\
\left.\widetilde M_{(s)}^{\to y}\right|_{A_O^s A_I^y\alpha_2^{(0)}}
&=\widehat M_{\varnothing}^{\to s}*\widehat M_{(s)}^{\to y}.
\label{eq:pf-firstslot-initial-cross}
\end{align}
If $1\le q\le N$, $\pi=(\pi_1,\ldots,\pi_q)$ does not contain $s$, and $y\notin\{s,\pi_1,\ldots,\pi_q\}$, then $\widetilde M_{(s,\pi)}^{\to y}\in\Pos(A_O^{\pi_q}\eta_{q+1}A_I^y\eta_{q+2})$ has only three types of nonzero blocks:
\begin{align}
\left.\widetilde M_{(s,\pi)}^{\to y}\right|_{A_O^{\pi_q}\alpha_q\bar A_{O,q+1}^s A_I^y\alpha_{q+1}\bar A_{O,q+2}^s}
&=\widehat M_\pi^{\to y}\otimes
|\id\rangle\!\rangle\langle\!\langle\id|_{\bar A_{O,q+1}^s\bar A_{O,q+2}^s},
\label{eq:pf-firstslot-direct}\\
\left.\widetilde M_{(s,\pi)}^{\to y}\right|_{A_O^{\pi_q}\alpha_q\bar A_{O,q+1}^s A_I^y\alpha_{q+2}^{(q)}}
&=\widehat M_\pi^{\to s}*\widehat M_{(\pi,s)}^{\to y},
\label{eq:pf-firstslot-cross}\\
\left.\widetilde M_{(s,\pi)}^{\to y}\right|_{A_O^{\pi_q}\alpha_{q+1}^{(r)}A_I^y\alpha_{q+2}^{(r)}}
&=\widehat M_{\pi[s.r]}^{\to y}
\;(0\le r\le q-1).
\label{eq:pf-firstslot-passed}
\end{align}
In Eq.~\eqref{eq:pf-firstslot-cross}, the unprojected $A_O^s$ is identified with the stored copy $\bar A_{O,q+1}^s$ through the fixed isomorphism. If $|\pi|=N+1$, the nonzero blocks of the final map are
\begin{align}
\left.\widetilde M_{(s,\pi)}^{\rm fin}\right|_{A_O^{\pi_{N+1}}\alpha_{N+1}\bar A_{O,N+2}^s}
&=\widehat M_\pi^{\to s}*\widehat M_{(\pi,s)}^{\rm fin},
\nonumber\\
\left.\widetilde M_{(s,\pi)}^{\rm fin}\right|_{A_O^{\pi_{N+1}}\alpha_{N+2}^{(r)}}
&=\widehat M_{\pi[s.r]}^{\rm fin}
\;(0\le r\le N).
\label{eq:pf-firstslot-final}
\end{align}

On a pre-$s$ input $\rho_s\otimes Z$, summing the new instrument combines the original alternatives $y\ne s$ with the branch through $s$. Trace preservation at $\pi$ and $(\pi,s)$ therefore gives $(\Tr\rho_s)\bigl(\sum_{y\notin\{s,\pi_1,\ldots,\pi_q\}}\Tr[Z*\widehat M_\pi^{\to y}]+\Tr[Z*\widehat M_\pi^{\to s}]\bigr)=\Tr(\rho_s\otimes Z)$. Eq.~\eqref{eq:pf-firstslot-final} gives the terminal case. Post-$s$ sectors inherit the original maps, so linearity gives trace preservation on the effective-input cone.

Denote the branch contribution of a complete history $\sigma$ in the original circuit by
\[
\widehat W_\sigma:=
\widehat M_{\varnothing}^{\to\sigma_1}*
\widehat M_{(\sigma_1)}^{\to\sigma_2}*\cdots *
\widehat M_{(\sigma_1,\ldots,\sigma_{N+1})}^{\to\sigma_{N+2}}*
\widehat M_\sigma^{\rm fin}.
\]
Fix a permutation not containing $s$, denoted by $\pi\in\Pi_{\widehat\N\setminus\{s\}}$. The sector paths of the complete history $(s,\pi)$ in the new circuit are uniquely labeled by the original position of $s$, namely $r=0,\ldots,N+1$. Repeatedly applying
\[
\bigl(X\otimes|\id\rangle\!\rangle\langle\!\langle\id|_{UV}\bigr)*Y_V=X*Y_U
\]
eliminates the storage wires and gives the branch contribution $\widetilde W_{(s,\pi).r}=\widehat W_{\pi[s.r]}$ for each path. Thus
\begin{align}
\widetilde W_{(s,\pi)}
&=\sum_{r=0}^{N+1}\widehat W_{\pi[s.r]},
\label{eq:pf-firstslot-branch-r-decomposition}\\
\widetilde W
&=\sum_{\pi\in\Pi_{\widehat\N\setminus\{s\}}}
  \sum_{r=0}^{N+1}\widehat W_{\pi[s.r]}
=\sum_{\sigma\in\Pi_{\widehat\N}}\widehat W_\sigma
=\widehat W.
\label{eq:pf-firstslot-reorder-process-invariant}
\end{align}
The last equality uses the bijection $(\pi,r)\mapsto\pi[s.r]$. Hence moving a slot with a one-dimensional input to the first position leaves the process matrix unchanged.

\smallskip
\noindent\emph{(b) If $A_O^t\simeq\mathbb C$, then $t$ can be moved to the last position in every nonzero complete history.}
For each $j=2,\ldots,N+2$, take a copy $\bar A_{I,j}^t\simeq A_I^t$, and again take $\alpha_q^{(r)}\simeq\alpha_q$. The stage memories of the new circuit are
\begin{equation}
\zeta_0:=\mathbb C,
\;
\zeta_q:=\alpha_q\oplus
\bigoplus_{r=0}^{q-1}(\bar A_{I,q+1}^t\alpha_{q+1}^{(r)})
\;(1\le q\le N+1),
\;
\zeta_{N+2}:=\alpha_{N+2}\oplus\bigoplus_{r=0}^{N}\mathbb C^{(r)}.
\label{eq:pf-lastslot-zeta-new}
\end{equation}
Here $\mathbb C^{(r)}$ are mutually orthogonal one-dimensional sectors labeled by $r$.
The first term is the pre-$t$ sector. The $r$th direct-sum term is the post-$t$ sector with index $r$ and stores $A_I^t$, whose slot has not yet actually been called. These pre/post labels refer to the original circuit history. The reordered circuit has a different call order. For each $y\ne t$, the nonzero root blocks are
\begin{align}
\left.\widetilde N_{\varnothing}^{\to y}\right|_{A_I^y\alpha_1}
&=\widehat M_{\varnothing}^{\to y},
\nonumber\\
\left.\widetilde N_{\varnothing}^{\to y}\right|_{A_I^y\bar A_{I,2}^t\alpha_2^{(0)}}
&=\widehat M_{\varnothing}^{\to t}*\widehat M_{(t)}^{\to y}.
\label{eq:pf-lastslot-initial}
\end{align}
If $1\le q\le N$, $\pi$ does not contain $t$, and $y\notin\{t,\pi_1,\ldots,\pi_q\}$, then the nonzero blocks of $\widetilde N_\pi^{\to y}$ are
\begin{align}
\left.\widetilde N_\pi^{\to y}\right|_{A_O^{\pi_q}\alpha_q A_I^y\alpha_{q+1}}
&=\widehat M_\pi^{\to y},
\label{eq:pf-lastslot-direct}\\
\left.\widetilde N_\pi^{\to y}\right|_{A_O^{\pi_q}\alpha_q A_I^y\bar A_{I,q+2}^t\alpha_{q+2}^{(q)}}
&=\widehat M_\pi^{\to t}*\widehat M_{(\pi,t)}^{\to y},
\label{eq:pf-lastslot-cross}\\
\left.\widetilde N_\pi^{\to y}\right|_{A_O^{\pi_q}\bar A_{I,q+1}^t\alpha_{q+1}^{(r)}A_I^y\bar A_{I,q+2}^t\alpha_{q+2}^{(r)}}
&=|\id\rangle\!\rangle\langle\!\langle\id|_{\bar A_{I,q+1}^t\bar A_{I,q+2}^t}
\otimes\widehat M_{\pi[t.r]}^{\to y}
\;(0\le r\le q-1).
\label{eq:pf-lastslot-passed}
\end{align}
If $|\pi|=N+1$, the actual call to $t$ occurs only at this stage. Take
$\widetilde N_\pi^{\to t}\in\Pos(A_O^{\pi_{N+1}}\zeta_{N+1}A_I^t\zeta_{N+2})$, with nonzero blocks
\begin{align}
\left.\widetilde N_\pi^{\to t}\right|_{A_O^{\pi_{N+1}}\alpha_{N+1}A_I^t\alpha_{N+2}}
&=\widehat M_\pi^{\to t},
\nonumber\\
\left.\widetilde N_\pi^{\to t}\right|_{A_O^{\pi_{N+1}}\bar A_{I,N+2}^t\alpha_{N+2}^{(r)}A_I^t\mathbb C^{(r)}}
&=|\id\rangle\!\rangle\langle\!\langle\id|_{\bar A_{I,N+2}^tA_I^t}
\otimes\widehat M_{\pi[t.r]}^{\rm fin}
\;(0\le r\le N).
\label{eq:pf-lastslot-finalstep}
\end{align}
After calling $t$, take the final internal map on $\alpha_{N+2}$ to be $\widehat M_{(\pi,t)}^{\rm fin}$ and the map on each one-dimensional sector to be the scalar identity.

In the pre-$t$ sector, the new outcome sum combines the original alternatives $y\ne t$ with the branch through $t$, preserving the trace. In a post-$t$ sector, let $Y$ be an effective input and $\mathcal E$ the corresponding original outcome-summed map. Its marginal is an effective input of the original circuit with the trace map inserted at $t$, so $\Tr[(\mathcal I_{\bar A_{I,q+1}^t}\otimes\mathcal E)(Y)]=\Tr[\mathcal E(\Tr_{\bar A_{I,q+1}^t}Y)]=\Tr Y$.

Fix a permutation not containing $t$, denoted by $\pi$. Connecting the identity storage wires segment by segment gives $\widetilde W_{(\pi,t).r}=\widehat W_{\pi[t.r]}$. The total branch contribution for the fixed actual call sequence $(\pi,t)$ is therefore $\sum_{r=0}^{N+1}\widehat W_{\pi[t.r]}$. Summing over all $\pi$ and again using the bijection $(\pi,r)\mapsto\pi[t.r]$ shows that the reordered process matrix is still $\widehat W$.

\smallskip
\noindent\emph{(c) Reorder the $P,F$ parties to the first and last positions and establish the equivalence of the two representations.}
First take $s=\mathsf P$. Since $A_I^{\mathsf P}=\mathbb C$, step (a) moves $\mathsf P$ to the first position in every nonzero complete call sequence. Next take $t=\mathsf F$. Since $A_O^{\mathsf F}=\mathbb C$, step (b) moves $\mathsf F$ to the last position while preserving the relative order of all other labels. Thus the same $\widehat W$ has a QC-CC realization whose nonzero complete call sequences all have the form of Eq.~\eqref{eq:pf-endpoint-order-sequence}.

Consider this realization with reordered endpoints. Step (a) gives a one-dimensional root memory and makes $\mathsf P$ the only initial outcome. In step (b), the additional root block associated with an initial call to $\mathsf F$ is therefore zero. Restricting the root memory to its occupied one-dimensional sector gives $\widehat M_{\varnothing}^{\to\mathsf P}=1$. Denote the internal CJ operators and stage memories by $\widehat M$ and $\alpha_q$, with $\alpha_1\simeq\mathbb C$. Removing the two one-dimensional ports of $\mathsf P,\mathsf F$ defines a $P\to F$ circuit by
\begin{align}
M_{\varnothing}^{\to k}&:=\widehat M_{(\mathsf P)}^{\to k},
\nonumber\\
M_h^{\to j}&:=\widehat M_{(\mathsf P,h)}^{\to j},
\nonumber\\
M_\pi^{\rm fin}&:=\widehat M_{(\mathsf P,\pi)}^{\to\mathsf F}*
\widehat M_{(\mathsf P,\pi,\mathsf F)}^{\rm fin}.
\label{eq:pf-endpoint-open-maps}
\end{align}
These CJ operators act on $PA_I^k\alpha_2$, $A_O^{\ell(h)}\alpha_{|h|+1}A_I^j\alpha_{|h|+2}$, and $A_O^{\pi_N}\alpha_{N+1}F$, respectively. The root and intermediate trace-preservation conditions are inherited directly. At a complete history $\pi$, for an effective input $Z$, set $X:=Z*\widehat M_{(\mathsf P,\pi)}^{\to\mathsf F}$. This node has only the outcome $\mathsf F$, so $\Tr X=\Tr Z$. Moreover, the deterministic map inserted at $\mathsf F:F\to\mathbb C$ is $\Tr_F$, and the original final internal map preserves the trace on $\Tr_FX$. Hence
\[
\Tr[Z*M_\pi^{\rm fin}]
=\Tr[(\Tr_FX)*\widehat M_{(\mathsf P,\pi,\mathsf F)}^{\rm fin}]
=\Tr Z.
\]
This gives a valid $P\to F$ QC-CC. The link product of the internal maps along each complete call sequence equals the $(\mathsf P,\pi,\mathsf F)$ branch of the reordered circuit~\cite{wechs2021QuantumCircuitsClassical}. After removing the two one-dimensional ports, its process matrix is $W=\widehat W$, so $W\in\QCCC_N^{P\to F}$.

Conversely, if $W\in\QCCC_N^{P\to F}$, take a circuit from Definition~\ref{def:basic-qccc} and define
\begin{align}
\widehat M_{\varnothing}^{\to\mathsf P}&=1,
\nonumber\\
\widehat M_{(\mathsf P)}^{\to k}&=M_{\varnothing}^{\to k},
\nonumber\\
\widehat M_{(\mathsf P,h)}^{\to j}&=M_h^{\to j},
\nonumber\\
\widehat M_{(\mathsf P,\pi)}^{\to\mathsf F}&=M_\pi^{\rm fin},
\;
\widehat M_{(\mathsf P,\pi,\mathsf F)}^{\rm fin}=1.
\label{eq:pf-endpoint-ordinary-maps}
\end{align}
All other sequence blocks are zero. The original $M_\pi^{\rm fin}$ has output $F$, which becomes the input of party $\mathsf F$, and the additional memory is one-dimensional. The trace-preservation conditions at every stage are those of the original $P\to F$ circuit, and the link product for each sequence is unchanged. The resulting process matrix $\widehat W=W$ therefore has an $(N+2)$-party QC-CC realization, proving Eq.~\eqref{eq:pf-qccc-participant-equivalence}. For $N=1$, there are no intermediate sequences, and every formula above remains well defined.
\end{proof}

\begin{proposition}[Consistency between ECS for processes with nontrivial $P,F$
and the ECS definition for one-dimensional $P,F$]
\label{prop:pf-consistency-with-wechs}
Let $N\ge1$, let $W$ be a finite-dimensional deterministic $P\to F$ process matrix, and let $\widehat W$ denote the representation of $W$ obtained by treating
$P$ and $F$ as parties with one-dimensional ports, as in
Eq.~\eqref{eq:pf-participant-representation}.
The following three conditions are equivalent:
\begin{equation}
   W\in\ECS_N^{P\to F}
    \;\;\text{if and only if}\;\;
   W\in\QCCC_N^{P\to F}
    \;\;\text{if and only if}\;\;
   \widehat W\in\ECS_{N+2}.
   \label{eq:pf-unified-equivalence}
\end{equation}
In particular, the direct definition of ECS for processes with nontrivial $P,F$ used here agrees exactly, in the ordinary $(N+2)$-party representation of the same operator, with the ECS definition for
one-dimensional $P,F$.
\end{proposition}

\begin{proof}
The first equivalence is precisely the main theorem, Theorem~\ref{thm:pf-ECS-QCCC}. If
$W\in\QCCC_N^{P\to F}$, Proposition~\ref{prop:pf-qccc-endpoint-reordering} gives
$\widehat W\in\QCCC_{N+2}$. For this ordinary $(N+2)$-party process, Theorem~\ref{thm:pf-ECS-QCCC} with one-dimensional global systems gives
$\widehat W\in\ECS_{N+2}$.

Conversely, if $\widehat W\in\ECS_{N+2}$, then
Theorem~\ref{thm:pf-ECS-QCCC} with one-dimensional global systems gives
$\widehat W\in\QCCC_{N+2}$. Proposition~\ref{prop:pf-qccc-endpoint-reordering}
then reinterprets the first and last parties, each with a one-dimensional port, as the global past and global future, yielding
$W\in\QCCC_N^{P\to F}$.

\end{proof}

\section{Probabilistic ECS and probabilistic QC-CC}

\label{chap:probabilistic-ecs}

\providecommand{\pCS}{\mathrm{pCS}}
\providecommand{\pECS}{\mathrm{pECS}}
\providecommand{\pQCCC}{\mathrm{pQC\mbox{-}CC}}

Fix a nonempty finite outcome set \(\mathcal R\). Write an outcome-indexed process family whose sum over outcomes is deterministic as
\(\mathbf W=(W^{[r]})_{r\in\mathcal R}\), and a general unnormalized outcome-indexed process family as
\(\mathbf X=(X^{[r]})_{r\in\mathcal R}\). Below, we use \emph{outcome family} as shorthand for the complete outcome-indexed process family.

\begin{definition}[Superinstruments and outcome elements]
\label{def:prob-process-superinstrument}
An outcome-indexed process family of PSD operators $\mathbf W=(W^{[r]})_{r\in\mathcal R}$, with $W^{[r]}\in\Pos(P\otimes A_{IO}^{\N}\otimes F)$, is called an \emph{superinstrument} if its sum over outcomes $\overline W:=\sum_{r\in\mathcal R}W^{[r]}$ is a deterministic $P\to F$ process matrix. Each $W^{[r]}$ is called an \emph{outcome element} or \emph{unnormalized process element} of the superinstrument and is not required to satisfy deterministic normalization individually. Given Choi operators $C_1,\ldots,C_N$ of external CP maps, denote the Choi operator of the global CP map for outcome $r$ by $\Lambda^{[r]}:=(C_1\otimes\cdots\otimes C_N)*W^{[r]}$. When all $C_k$ are CPTP, $\sum_r\Lambda^{[r]}=(C_1\otimes\cdots\otimes C_N)*\overline W$ is the Choi operator of a CPTP map from $P$ to $F$.
\end{definition}

\begin{definition}[Probabilistic QC-CC]
\label{def:pqccc-operational}
For $N=0$, $\pQCCC_0^{P\to F}$ consists of all quantum instruments from $P$ to $F$ with finitely many outcomes, namely, all outcome-indexed process families satisfying $W^{[r]}\succeq0$ whose sum $\sum_rW^{[r]}$ is the Choi operator of a CPTP map.

Let \(N\ge1\). A \emph{probabilistic quantum circuit with classical control of causal order} has the same initial and intermediate instruments as the deterministic QC-CC of Definition~\ref{def:basic-qccc}, while the final map after each complete history \(\pi\in\Pi_{\N}\) is replaced by a quantum instrument $\{\mathcal M_{\pi}^{\mathrm{fin}[r]}\}_{r\in\mathcal R}$. The outcome sum of every initial, intermediate, and final instrument obeys the corresponding effective-input trace-preservation condition of Definition~\ref{def:basic-qccc}~\cite{wechs2021QuantumCircuitsClassical}.

Using the same CJ notation, the joint contribution of a complete history
\(\pi=(\pi_1,\ldots,\pi_N)\) and a final outcome \(r\) is
\begin{equation}
   T_{\pi}^{[r]}
   :=
   M_{\varnothing}^{\to\pi_1}
   *M_{(\pi_1)}^{\to\pi_2}
   *\cdots*
   M_{(\pi_1,\ldots,\pi_{N-1})}^{\to\pi_N}
   *M_{\pi}^{\mathrm{fin}[r]},
\label{eq:pqccc-terminal-contribution}
\end{equation}
For $N=1$, there are no intermediate factors, and the expression reduces to $M_{\varnothing}^{\to k}*M_{(k)}^{\mathrm{fin}[r]}$.
The outcome element corresponding to \(r\) is
\begin{equation}
   W^{[r]}=\sum_{\pi\in\Pi_{\N}}T_{\pi}^{[r]}.
   \label{eq:pqccc-outcome-sum}
\end{equation}
The class of all outcome-indexed process families realizable by such circuits is denoted by
\(\pQCCC_N^{P\to F}\).
\end{definition}
\begin{lemma}[]
\label{lem:pqccc-outcome-lift}
Let a family of PSD operators \(T_\pi^{[r]}\) satisfy condition (ii) of
Proposition~\ref{prop:pqccc-history-characterization}.
Take an output register of dimension \(|\mathcal R|\), denoted by \(F'\), fix an orthogonal basis
\(\{\ket r^{F'}\}_{r\in\mathcal R}\), and define
\begin{equation}
   \widetilde W
   :=\sum_{r\in\mathcal R}W^{[r]}\otimes\proj r^{F'},
   \qquad
   \widetilde T_\pi
   :=\sum_{r\in\mathcal R}T_\pi^{[r]}\otimes\proj r^{F'}.
   \label{eq:pqccc-classical-register}
\end{equation}
Then \(\widetilde W\) is a deterministic QC-CC process with global future \(F\otimes F'\).
Apply the deterministic QC-CC realization construction, resolve the final map after each complete history according to
\(\{\proj r^{F'}\}_r\), and discard \(F'\). This yields a common set of root and intermediate internal instruments and final CP instrument elements
\(\{M_\pi^{\mathrm{fin}[r]}\}_r\), with outcome elements exactly equal to \(T_\pi^{[r]}\).
\end{lemma}

\begin{proof}
Every $\widetilde T_\pi$ is PSD because $T_\pi^{[r]}\succeq0$ and the projectors $\proj r^{F'}$ are orthogonal. Set $T_\pi:=\sum_rT_\pi^{[r]}$. For each nonempty history $h=(k_1,\ldots,k_n)$, define $T_h:=\sum_{\pi\succeq h}T_\pi$ and $\widetilde T_h:=\sum_{\pi\succeq h}\widetilde T_\pi$. Tracing out the classical register gives $\Tr_{F'}\widetilde T_h=T_h$. Hence future depolarization obeys
\begin{equation}
\mathcal D_{A_{IO}^{\N\setminus\{k_1,\ldots,k_n\}}FF'}(\widetilde T_h)
=\frac{1}{d_{F'}}\mathcal D_{A_{IO}^{\N\setminus\{k_1,\ldots,k_n\}}F}(T_h)\otimes\id^{F'}.
\end{equation}
The prefix constraint on $T_h$ therefore implies
\begin{equation}
(\mathcal I-\mathcal D_{A_O^{\ell(h)}})
\mathcal D_{A_{IO}^{\N\setminus\{k_1,\ldots,k_n\}}FF'}(\widetilde T_h)=0.
\end{equation}
Moreover,
\begin{equation}
\Tr_{A_{IO}^{\N}FF'}\widetilde W
=\Tr_{A_{IO}^{\N}F}\overline W
=D_O(\N)\id^P.
\end{equation}
Thus the lifted decomposition satisfies both the homogeneous history constraints and deterministic normalization. Since $\Tr_{F'}\widetilde W=\overline W$, inserting arbitrary local CPTP maps yields a CP map whose output trace equals that of the deterministic process $\overline W$, namely $\id^P$. Hence $\widetilde W$ is a deterministic $P\to FF'$ process. Apply the realization construction in Proposition~\ref{prop:basic-pf-qccc} to these specified history operators and terminal components $\widetilde T_\pi$. It gives a deterministic $P\to F\otimes F'$ QC-CC circuit whose contribution for each complete history is $\widetilde T_\pi$, with common root and intermediate instruments and final maps with output $FF'$. After each final map, measuring $F'$ with $\proj r^{F'}$ is the link product with that projector in the chosen CJ basis and recovers $T_\pi^{[r]}$. Summing the CP outcome elements over $r$ discards $F'$, so they form a final quantum instrument. This is the construction in Proposition 13 and Appendix B.2.d of Wechs et al.~\cite{wechs2021QuantumCircuitsClassical}.
\end{proof}
\begin{proposition}[Characterization of \(\pQCCC\)]
\label{prop:pqccc-history-characterization}
Let \(N\ge1\), and let
\(\mathbf W=(W^{[r]})_{r\in\mathcal R}\) be a family of PSD operators. The following are equivalent:
\begin{enumerate}[label=\textup{(\roman*)}]
\item \(\mathbf W\in\pQCCC_N^{P\to F}\),
\item The sum over outcomes $\overline W:=\sum_{r\in\mathcal R}W^{[r]}$
is a deterministic \(P\to F\) process, and there exist PSD operators   $T_{\pi}^{[r]}\succeq0
    \;
   (\pi\in\Pi_{\N},\ r\in\mathcal R),$
satisfying
\begin{equation}
   W^{[r]}=\sum_{\pi\in\Pi_{\N}}T_{\pi}^{[r]}
    \;(r\in\mathcal R).
   \label{eq:pqccc-history-outcomes}
\end{equation}
Define
\[
   T_{\pi}:=\sum_{r\in\mathcal R}T_{\pi}^{[r]},
   \qquad
   T_h:=\sum_{\pi\succeq h}T_{\pi}
   =\sum_{r\in\mathcal R}\sum_{\pi\succeq h}T_\pi^{[r]}.
\]
The common history constraints apply to \(T_h\), after summing over final outcomes. The history partial sums for a fixed \(r\) need not satisfy these constraints individually.
Then, for every nonempty history \(h\),
\begin{equation}
   (\mathcal I-\mathcal D_{A_O^{\ell(h)}})
   \mathcal D_{A_{IO}^{\N\setminus\{k_1,\ldots,k_n\}}F}T_h=0.
   \label{eq:pqccc-history-constraint}
\end{equation}
\end{enumerate}
\end{proposition}
\begin{proof}
Assume (i), and let $T_\pi^{[r]}$ be the complete-history outcome contributions defined in Eq.~\eqref{eq:pqccc-terminal-contribution}. Summing over $r$ replaces each final instrument by its CPTP sum and yields the deterministic QC-CC process $\overline W$, with complete-order components $T_\pi:=\sum_rT_\pi^{[r]}$. Proposition~\ref{prop:basic-pf-qccc} then gives Eq.~\eqref{eq:pqccc-history-constraint}, proving (ii).

Conversely, Lemma~\ref{lem:pqccc-outcome-lift} realizes every family satisfying (ii) with common root and intermediate instruments and an outcome-resolved final instrument, proving (i).
\end{proof}

For any outcome family \(\mathbf X=(X^{[r]})_{r\in\mathcal R}\), write \(\overline X:=\sum_{r\in\mathcal R}X^{[r]}\). If party \(k\) undergoes a quantum operation with Choi matrix \(C_k\), denote outcome-by-outcome conditioning by \(\mathbf X_{|C_k}:=((X^{[r]})_{|C_k})_{r\in\mathcal R}\).

\begin{definition}[Probabilistic causally separable cone]
\label{def:pcs-recursive}
Define \(\pCS_N^{+,P\to F}\) recursively.

For \(N=0\), set
\begin{equation}
   \pCS_0^{+,P\to F}
   :=\bigl\{\mathbf X=(X^{[r]})_{r\in\mathcal R}\mid
       X^{[r]}\in\Pos(P\otimes F)\text{ for all }r\in\mathcal R\bigr\}.
   \label{eq:pcs-base}
\end{equation}

For \(N\ge1\), a family
\(\mathbf X=(X^{[r]})_{r\in\mathcal R}\)
belongs to \(\pCS_N^{+,P\to F}\) if and only if there exist PSD operators
\(X_{(k)}^{[r]}\succeq0\)
(\(k\in\N,r\in\mathcal R\)) such that
\begin{align}
   X^{[r]}&=\sum_{k\in\N}X_{(k)}^{[r]}
   ,\;(r\in\mathcal R),
   \label{eq:pcs-first-split}\\
   \overline X_{(k)}
   :=\sum_{r\in\mathcal R}X_{(k)}^{[r]}
   &\in\mathsf{Proc}_{N}^{+,P\to F}
   ,\;(k\in\N),
   \label{eq:pcs-branch-valid}\\
   \bigl((X_{(k)}^{[r]})_{|C_k}\bigr)_{r\in\mathcal R}
   &\in\pCS_{N-1}^{+,P\to F}
   ,\;\text{for every quantum operation $C_k$}.
   \label{eq:pcs-recursion}
\end{align}
The subscript \(N-1\) on the right records the remaining slots, labeled by \(\N\setminus\{k\}\). The same \(P,F\) and outcome set \(\mathcal R\) are retained throughout. Input extensions are included in the corresponding parties' input spaces.
\end{definition}

Only the outcome-summed branch $\overline X_{(k)}$ is required to be a valid unnormalized process; the individual $X_{(k)}^{[r]}$ need not be valid separately.

\begin{definition}[Probabilistic extensible causal separability]
\label{def:pecs-original}
A family of PSD operators
\(\mathbf X=(X^{[r]})_{r\in\mathcal R}\)
is called an \emph{unnormalized probabilistic ECS family}, written \(\mathbf X\in\pECS_N^{+,P\to F}\), if, for any finite-dimensional input ancillas \(E_1,\ldots,E_N\) and any shared state satisfying \(\rho\in\Pos(E_1\otimes\cdots\otimes E_N)\) and \(\Tr\rho=1\),
\begin{equation}
   \mathbf X\otimes\rho
   :=\bigl(X^{[r]}\otimes\rho\bigr)_{r\in\mathcal R}
   \in\pCS_N^{+,P\to F},
   \label{eq:pecs-extension}
\end{equation}
where the right-hand side is understood on the extended input slots.

If, in addition,
\(\overline X=\sum_rX^{[r]}\)
satisfies \(\Theta_{\overline X}=\id^P\), denote the family by $\mathbf W$, call it a \emph{probabilistic extensibly causally separable superinstrument}, and write \(\mathbf W\in\pECS_N^{P\to F}\).
\end{definition}

For \(N\ge1\), choosing trivial input ancillas in Eq.~\eqref{eq:pecs-extension} and summing Eq.~\eqref{eq:pcs-branch-valid} gives \(\overline X\in\mathsf{Proc}_N^{+,P\to F}\). Thus the condition \(\Theta_{\overline X}=\id^P\) selects families whose outcome sum is deterministic.

For \(N=0\), there are no ancillary input extensions or intermediate recursion. Thus
\(\pECS_0^{P\to F}\) consists precisely of all PSD outcome families whose sum over outcomes is a CPTP Choi operator from \(P\) to \(F\). This is exactly
\(\pQCCC_0^{P\to F}\) in Definition~\ref{def:pqccc-operational}.

\begin{lemma}[Closure of the probabilistic causally separable cone and discarding of input ancillas]
\label{lem:pcs-closure-discard}
For any \(N\ge0\):
\begin{enumerate}[label=\textup{(\roman*)}]
\item \(\pCS_N^{+,P\to F}\) is a closed convex cone in the corresponding finite-dimensional space of outcome families,
\item If the input of each party \(j\) additionally contains an ancilla \(G_j\) and
\(\mathbf Y\in\pCS_N^{+,P\to F}\)
is understood on these extended inputs, then the outcome-by-outcome partial trace
\begin{equation}
   \Tr_{G_1\cdots G_N}\mathbf Y
   :=\bigl(\Tr_{G_1\cdots G_N}Y^{[r]}\bigr)_{r\in\mathcal R}
   \label{eq:pcs-discard}
\end{equation}
still belongs, with \(G_j\) removed, to \(\pCS_N^{+,P\to F}\).
\end{enumerate}
\end{lemma}

\begin{proof}
We proceed by induction on $N$. For $N=0$, Eq.~\eqref{eq:pcs-base} is a finite Cartesian product of PSD cones and is therefore a closed convex cone. Partial tracing preserves positivity outcome by outcome, so both claims hold.

Assume the result for $N-1$. Convexity follows from the linearity of the decomposition and conditioning maps and the convexity of the process cone and $\pCS_{N-1}^{+,P\to F}$. For closedness, let $\mathbf X_m\to\mathbf X$ with $\mathbf X_m\in\pCS_N^{+,P\to F}$, and choose witnesses $X_{m,(k)}^{[r]}$. Since $0\preceq X_{m,(k)}^{[r]}\preceq X_m^{[r]}$ and the outcome and first-level index sets are finite, these witnesses are bounded. A common subsequence converges to operators $X_{(k)}^{[r]}\succeq0$. The outcome sums, branch-validity constraints, and the finite homogeneous process constraints are preserved in the limit. For each fixed $k$ and quantum-operation CJ matrix $C_k$, conditioning is continuous, while $\pCS_{N-1}^{+,P\to F}$ is closed by induction. Intersecting the corresponding closed inverse images over all $C_k$ proves the recursive condition and hence closedness.

For the partial-trace statement, take a witness $Y_{(k)}^{[r]}$ for $\mathbf Y$, write $\mathbf Y_{(k)}:=(Y_{(k)}^{[r]})_{r\in\mathcal R}$, and define $\widetilde Y_{(k)}^{[r]}:=\Tr_{G_1\cdots G_N}Y_{(k)}^{[r]}$. Positivity and the outcome sums are preserved. The branch sum $\sum_rY_{(k)}^{[r]}$ is a valid process on the extended inputs; commuting the partial trace through the depolarizing maps and output projections, as in Eq.~\eqref{eq:pf-discard-proc}, shows that $\sum_r\widetilde Y_{(k)}^{[r]}$ is valid after the $G_j$ are removed. For a quantum-operation CJ matrix $C_k$ on the reduced input, lift it to $C_k^\uparrow:=C_k\otimes\id^{G_k}$. Here $\id^{G_k}$ is the CJ operator of the trace (discard) map on the retained input ancilla $G_k$, which has no output slot in the reduced process, so $\Tr_{A_O^k}C_k^\uparrow=(\Tr_{A_O^k}C_k)\otimes\id^{G_k}\preceq\id^{A_I^kE_kG_k}$. Hence it is an allowed quantum operation. Then conditioning and partial trace commute:
\[
\left(\Tr_{G_1\cdots G_N}\mathbf Y_{(k)}\right)_{|C_k}
=\Tr_{\otimes_{j\ne k}G_j}\left(\mathbf Y_{(k)}\right)_{|C_k^\uparrow}.
\]
The conditioned family on the right belongs to $\pCS_{N-1}^{+,P\to F}$ by the witness property, and the induction hypothesis preserves this property after tracing the remaining ancillas. Thus the partially traced outcome family belongs to $\pCS_N^{+,P\to F}$.
\end{proof}

\begin{lemma}[Extension-independent first-level decomposition for \(\pECS\)]
\label{lem:pecs-uniform}
Let $N\ge1$. If \(\mathbf X=(X^{[r]})_{r\in\mathcal R}\in\pECS_N^{+,P\to F}\), then there exists a set of PSD operators
\(X_{(k)}^{[r]}\) satisfying
Eqs.~\eqref{eq:pcs-first-split}--\eqref{eq:pcs-branch-valid} such that, for every finite-dimensional input extension, every normalized shared state \(\rho\), every \(k\), and every quantum operation at party \(k\) with Choi matrix \(C_k\),
\begin{equation}
   \bigl((X_{(k)}^{[r]}\otimes\rho)_{|C_k}\bigr)_{r\in\mathcal R}
   \in\pCS_{N-1}^{+,P\to F}.
   \label{eq:pecs-uniform-condition}
\end{equation}
\end{lemma}

\begin{proof}
First fix finite-dimensional input extensions and a pure shared state
\(\rho=|\psi\rangle\!\langle\psi|\). By Definition~\ref{def:pecs-original}, the extended family
\(\mathbf X\otimes\rho\) has a \(\pCS\) first-level witness \(Y_{(k)}^{[r]}\). For every \(r,k\), we have \(0\preceq Y_{(k)}^{[r]}\preceq X^{[r]}\otimes|\psi\rangle\!\langle\psi|\). By the inheritance of PSD support in Lemma~\ref{lem:pf-support}, there is a unique operator \(X_{(k)}^{[r,\psi]}\succeq0\) on the original space such that \(Y_{(k)}^{[r]}=X_{(k)}^{[r,\psi]}\otimes|\psi\rangle\!\langle\psi|\) and \(X^{[r]}=\sum_kX_{(k)}^{[r,\psi]}\). Summing over outcomes and tracing out all input ancillas gives \(\sum_rX_{(k)}^{[r,\psi]}\in\mathsf{Proc}_{N}^{+,P\to F}\).
The extended recursive condition then reduces to Eq.~\eqref{eq:pecs-uniform-condition} for these original-space operators and the fixed pure state.

For a mixed shared state, take a finite-dimensional purification and attach the purifying system to any one input. Apply the preceding argument to the purified state, then discard the purifying system using Lemma~\ref{lem:pcs-closure-discard}(ii). Thus, for each fixed extension state, at least one set of first-level operators on the original space satisfies all required conditions.

To obtain a common decomposition for all extension states, consider the finite-dimensional real vector space
   $\mathsf X_N^{\mathcal R}:=
   \prod_{r\in\mathcal R}\prod_{k\in\N}
   \Herm(P\otimes A_{IO}^{\N}\otimes F).$
For a fixed family $\mathbf X$, define
\begin{equation}
   \mathfrak K_{\mathbf X}:=
\left\{(Z_{(k)}^{[r]})_{k,r}\in\mathsf X_N^{\mathcal R}:
Z_{(k)}^{[r]}\succeq0,\ 
\sum_k Z_{(k)}^{[r]}=X^{[r]}\ \text{for every }r\right\}.
\end{equation}
This set is closed and bounded: every component satisfies
$0\preceq Z_{(k)}^{[r]}\preceq X^{[r]}$. It is therefore compact in the finite-dimensional space $\mathsf X_N^{\mathcal R}$.
Write $E=(E_1,\ldots,E_N)$ for an extension, with each $E_j=\mathbb C^{e_j}$ and $e_j\ge1$. For each normalized state $\rho$ on $\bigotimes_j E_j$, let
\begin{equation}
\begin{aligned}
\mathfrak F_{E,\rho}:=\bigl\{(Z_{(k)}^{[r]})_{k,r}\in\mathfrak K_{\mathbf X}:\
&\sum_r Z_{(k)}^{[r]}\in\mathsf{Proc}_N^{+,P\to F}\quad\text{for every }k,\\
&\bigl((Z_{(k)}^{[r]}\otimes\rho)_{|C_k}\bigr)_{r\in\mathcal R}
\in\pCS_{N-1}^{+,P\to F}
\quad\text{for every }k,C_k\bigr\}.
\end{aligned}
\end{equation}
Here $\rho$ and the extension spaces are fixed, and $C_k$ ranges over quantum-operation CJ matrices on $A_I^kE_kA_O^k$. Choosing the spaces $\mathbb C^{e_j}$ covers all finite-dimensional extensions after a choice of basis. The pure-state and purification arguments establish nonemptiness for each $(E,\rho)$. Process validity is closed, and each conditioning constraint is the inverse image of the closed cone $\pCS_{N-1}^{+,P\to F}$ under a continuous linear map. Their intersection over all $k,C_k$ is closed. Hence $\mathfrak F_{E,\rho}$ is a closed subset of the fixed compact set $\mathfrak K_{\mathbf X}$.

Take any finite collection $(E_t,\rho_t)$, $t=1,\ldots,m$. Give party $j$ the input ancilla $\bigotimes_t E_{j,t}$ and use the state $\bigotimes_t\rho_t$. A first-level decomposition for this joint extension gives one tuple $(Z_{(k)}^{[r]})_{k,r}$ on the original space. For a fixed $t$, lift a quantum operation $C_{k,t}$ to
$C_{k,t}\otimes\bigotimes_{u\ne t}\id^{E_{k,u}}$ and trace out the ancillary factors with $u\ne t$ at the remaining parties. Lemma~\ref{lem:pcs-closure-discard}(ii) gives
$\bigl((Z_{(k)}^{[r]}\otimes\rho_t)_{|C_{k,t}}\bigr)_{r\in\mathcal R}
\in\pCS_{N-1}^{+,P\to F}.$
Thus the same tuple belongs to every $\mathfrak F_{E_t,\rho_t}$. The closed sets $\mathfrak F_{E,\rho}$ have the finite intersection property in $\mathfrak K_{\mathbf X}$, so their total intersection is nonempty. Any tuple in this intersection satisfies Eq.~\eqref{eq:pecs-uniform-condition} for all finite-dimensional input extensions and shared states.

\end{proof}

\begin{proposition}[Equivalence of probabilistic ECS and \(\pQCCC\)]
\label{prop:pecs-pqccc-equivalence}
For any finite \(N\), finite-dimensional \(P,F,A_I^k,A_O^k\),    $\pECS_N^{P\to F}
   =
   \pQCCC_N^{P\to F}.$
\end{proposition}

\begin{proof}
The zero-slot case follows directly from Definitions~\ref{def:pcs-recursive}--\ref{def:pecs-original} and the subsequent discussion of $N=0$. Henceforth assume \(N\ge1\).

\medskip
\noindent\emph{Part I: \(\pQCCC_N^{P\to F}\subseteq\pECS_N^{P\to F}\).}

Fix $\mathbf W\in\pQCCC_N^{P\to F}$, let $T_\pi^{[r]}$ be the outcome-resolved complete-order components from Proposition~\ref{prop:pqccc-history-characterization}, and fix arbitrary finite-dimensional input ancillas $E_1,\ldots,E_N$ and a normalized shared state $\rho$.

For any history prefix $h=(k_1,\ldots,k_s)$, possibly empty, and any sequence of quantum-operation Choi matrices $\boldsymbol C_h=(C_{k_1},\ldots,C_{k_s})$, define
\begin{equation}
   Y_{h,\boldsymbol C_h}^{[r]}
   :=
   \left(
      \sum_{\pi\succeq h}T_{\pi}^{[r]}\otimes\rho
   \right)_{|\boldsymbol C_h}.
   \label{eq:pqccc-pecs-future-family}
\end{equation}
We prove by induction on $m:=|\N\setminus\{k_1,\ldots,k_s\}|$ the stronger claim
\begin{equation}
   \mathbf Y_{h,\boldsymbol C_h}
   :=(Y_{h,\boldsymbol C_h}^{[r]})_{r\in\mathcal R}
   \in\pCS_m^{+,P\to F}.
   \label{eq:pqccc-pecs-strong-claim}
\end{equation}

If \(m=0\), each \(Y_{h,\boldsymbol C_h}^{[r]}\) is a PSD operator on \(P\otimes F\), so the claim follows from Eq.~\eqref{eq:pcs-base}.

Assume the claim holds for all cases with at most \(m-1\) remaining slots, and let \(|\N\setminus\{k_1,\ldots,k_s\}|=m\ge1\). For each
\(j\in\N\setminus\{k_1,\ldots,k_s\}\), define
\begin{equation}
   Y_{h\to j,\boldsymbol C_h}^{[r]}
   :=
   \left(
      \sum_{\pi\succeq(h,j)}T_{\pi}^{[r]}\otimes\rho
   \right)_{|\boldsymbol C_h}.
   \label{eq:pqccc-pecs-next-branch}
\end{equation}
For each outcome, \(Y_{h,\boldsymbol C_h}^{[r]}=\sum_{j\in\N\setminus\{k_1,\ldots,k_s\}}Y_{h\to j,\boldsymbol C_h}^{[r]}\), and every component is PSD. For fixed $j$, insert arbitrary CPTP CJ operators in all $m$ current slots and sum over the final outcome $r$ and all future classical histories extending $(h,j)$. The resulting internal instruments are trace preserving on the effective-input cone at each node. Therefore the input-side trace of the resulting \(P\to F\) map depends only on the fixed past $(h,\boldsymbol C_h)$ and on the choice of the next branch $j$, and is independent of the future CPTP insertions. By Proposition~\ref{prop:pf-validity}, this trace-independence together with positivity implies $\sum_{r\in\mathcal R}Y_{h\to j,\boldsymbol C_h}^{[r]}
   \in\mathsf{Proc}_{m}^{+,P\to F}.$

At party \(j\), fix an arbitrary quantum operation with Choi matrix \(C_j\). By linearity of conditioning,
\begin{equation}
   \bigl((Y_{h\to j,\boldsymbol C_h}^{[r]})_{|C_j}\bigr)_{r\in\mathcal R}
   =
   \mathbf Y_{(h,j),(\boldsymbol C_h,C_j)}.
   \label{eq:pqccc-pecs-conditional-branch}
\end{equation}
The right-hand side has only \(m-1\) remaining slots and hence belongs to
\(\pCS_{m-1}^{+,P\to F}\) by the induction hypothesis. Thus
Eqs.~\eqref{eq:pqccc-pecs-next-branch}--\eqref{eq:pqccc-pecs-conditional-branch}
satisfy every condition in Definition~\ref{def:pcs-recursive}, proving
Eq.~\eqref{eq:pqccc-pecs-strong-claim} for \(m\).

Taking the empty history \(h=\varnothing\) with no conditioning gives \((W^{[r]}\otimes\rho)_{r\in\mathcal R}\in\pCS_N^{+,P\to F}\). Since the input ancillas and shared state are arbitrary,
\(\mathbf W\in\pECS_N^{+,P\to F}\). Moreover, the sum over outcomes of a \(\pQCCC\) family is a deterministic process matrix, so
\(\mathbf W\in\pECS_N^{P\to F}\).

\medskip
\noindent\emph{Part II: \(\pECS_N^{P\to F}\subseteq\pQCCC_N^{P\to F}\).}

We first prove the unnormalized history-construction statement: if
\(\mathbf X\in\pECS_N^{+,P\to F}\), then there exist
\begin{equation}
    T_{\pi}^{[r]}\succeq0
    \;
    (\pi\in\Pi_{\N},\ r\in\mathcal R)
    \label{eq:pecs-history-construction-components}
\end{equation}
such that $X^{[r]}=\sum_{\pi\in\Pi_{\N}}T_{\pi}^{[r]},$
and setting \(T_{\pi}:=\sum_rT_{\pi}^{[r]}, \;T_h:=\sum_{\pi\succeq h}T_{\pi}\), every nonempty history satisfies
\begin{equation}
    (\mathcal I-\mathcal D_{A_O^{\ell(h)}})
    \mathcal D_{A_{IO}^{\N\setminus\{k_1,\ldots,k_n\}}F}T_h=0.
   \label{eq:pecs-history-construction-prefix}
\end{equation}
We prove this by induction on \(N\).

For \(N=1\), denote the sole external label by \(k\). By Lemma~\ref{lem:pecs-uniform}, the only first-level component is
\(X_{(k)}^{[r]}=X^{[r]}\), and \(\overline X=\sum_rX^{[r]}\in\mathsf{Proc}_{1}^{+,P\to F}\). Set \(T_{(k)}^{[r]}:=X^{[r]}\). The validity constraint of the unnormalized process for the set \(\{k\}\) gives \((\mathcal I-\mathcal D_{A_O^k})\mathcal D_F\overline X=0\), proving the claim.

Let \(N\ge2\), and assume the claim holds for any outcome family on \(N-1\) slots. Apply Lemma~\ref{lem:pecs-uniform} to
\(\mathbf X\in\pECS_N^{+,P\to F}\) to obtain fixed first-level operators
\(X_{(k)}^{[r]}\). Fix \(k\) and write \(\overline X_{(k)}:=\sum_rX_{(k)}^{[r]}\). Use the maximally entangled projection constructed in Sec.~\ref{sec:pf-route} for party \(k\), retaining its notation $G_j,V,\mathcal V$ and $\beta$. The isometric embedding \(\mathcal V\) is independent of both the outcome \(r\) and any subsequent extension state, so the same embedding acts on every outcome element. 
Define
$Z_k^{[r]}:=\beta\mathcal V
X_{(k)}^{[r]}\mathcal V^\dagger$,
where $\beta>0$ is the constant in
Eq.~\eqref{eq:pf-teleport}.
Take arbitrary further input ancillas for the remaining \(N-1\) parties and any shared state \(\tau\). In Eq.~\eqref{eq:pecs-uniform-condition}, choose the shared input to be the tensor product of the entangled state and \(\tau\), and choose the same maximally entangled projection operation at party \(k\). Conditioning on the maximally entangled projection gives, for each outcome,
\(Z_k^{[r]}\otimes\tau\). Hence $\mathbf Z_k:=(Z_k^{[r]})_{r\in\mathcal R}
   \in\pECS_{N-1}^{+,P\to F},$
where the remaining slots are labeled by \(\N\setminus\{k\}\).

By the induction hypothesis, after prepending $k$ to each permutation of the remaining parties, there exist, for each outcome $r$ and each full permutation $\pi\succeq(k)$, PSD operators
\(Z_{\pi}^{[r]}\) such that
\begin{equation}
   Z_k^{[r]}
    =\sum_{\pi\succeq(k)}Z_{\pi}^{[r]},
    \label{eq:pecs-lower-components}
\end{equation}
and the remaining-process prefixes satisfy all history constraints after summing over outcomes. By \(0\preceq Z_{\pi}^{[r]}\preceq Z_k^{[r]}=\beta\mathcal V X_{(k)}^{[r]}\mathcal V^{\dagger}\) and Lemma~\ref{lem:pf-support}, each remaining-process component has the unique form
\begin{equation}
   Z_{\pi}^{[r]}
    =\beta\mathcal V X_{\pi}^{[r]}\mathcal V^{\dagger},
    \;
    X_{\pi}^{[r]}\succeq0.
    \label{eq:pecs-recovered-components}
\end{equation}
Summing Eq.~\eqref{eq:pecs-lower-components} and compressing by the adjoint isometry gives
\begin{equation}
   X_{(k)}^{[r]}
    =\sum_{\pi\succeq(k)}X_{\pi}^{[r]}.
    \label{eq:pecs-recovered-sum}
\end{equation}
For each full permutation $\pi\succeq(k)$, define $T_{\pi}^{[r]}:=X_{\pi}^{[r]}.$
Equations~\eqref{eq:pcs-first-split} and \eqref{eq:pecs-recovered-sum} give \(X^{[r]}=\sum_{\pi\in\Pi_{\N}}T_{\pi}^{[r]}\).

To verify the prefix constraints, the sum over outcomes for a length-one prefix \((k)\) is
\(\overline X_{(k)}\). Since
\(\overline X_{(k)}\in\mathsf{Proc}_{N}^{+,P\to F}\), its output-validity condition for
\(\{k\}\) gives
\begin{equation}
   (\mathcal I-\mathcal D_{A_O^k})
    \mathcal D_{A_{IO}^{\N\setminus\{k\}}F}\overline X_{(k)}=0.
   \label{eq:pecs-length-one-prefix}
\end{equation}

Next take any nonempty \(\N\setminus\{k\}\)-history
\(h=(u_1,\ldots,u_n)\). The induction hypothesis gives the remaining-process constraint on the extended slots,
\begin{equation}
   (\mathcal I-\mathcal D_{A_O^{u_n}})
   \mathcal D_{(\otimes_{j\in\N\setminus\{k,u_1,\ldots,u_n\}}A_{IO}^jG_j)F}
   \left(
      \sum_{r\in\mathcal R}
      \sum_{\pi\succeq(k,h)}Z_{\pi}^{[r]}
   \right)=0,
   \label{eq:pecs-routed-prefix}
\end{equation}
where each future slot is written explicitly as $A_{IO}^jG_j$. Substitute Eq.~\eqref{eq:pecs-recovered-components} and, following Lemma~\ref{lem:pf-compress}, compress by the adjoint of the isometry corresponding to the first remaining history label \(u_1\). The lemma gives a strictly positive proportionality factor, yielding
\begin{equation}
   (\mathcal I-\mathcal D_{A_O^{u_n}})
   \mathcal D_{A_{IO}^{\N\setminus\{k,u_1,\ldots,u_n\}}F}
   \left(
      \sum_{r\in\mathcal R}
      \sum_{\pi\succeq(k,h)}T_{\pi}^{[r]}
   \right)=0.
   \label{eq:pecs-long-prefix}
\end{equation}
Equation~\eqref{eq:pecs-length-one-prefix} covers all length-one histories, and Eq.~\eqref{eq:pecs-long-prefix} covers all longer histories, proving the unnormalized history-construction statement.

Finally, suppose that \(\mathbf W\in\pECS_N^{P\to F}\), so that its outcome sum \(\overline W\) is deterministic. The construction above yields
\(T_{\pi}^{[r]}\) satisfying
Eq.~\eqref{eq:pqccc-history-outcomes} and
Eq.~\eqref{eq:pqccc-history-constraint}. Proposition~\ref{prop:pqccc-history-characterization} therefore gives
\(\mathbf W\in\pQCCC_N^{P\to F}\).

\end{proof}

\end{document}